\documentclass{amsart}
\usepackage[utf8]{inputenc}
\usepackage[T2A,T1]{fontenc}
\usepackage[top=1in, bottom=1.5in, left=1.5in, right=1.5in]{geometry}
\usepackage{amsaddr}

\DeclareSymbolFont{cyrillic}{T2A}{cmr}{m}{n}
\DeclareMathSymbol{\El}{\mathalpha}{cyrillic}{203}

\usepackage[backend=biber,backref=true,style=alphabetic,giveninits=true]{biblatex} 

\usepackage[dvipsnames]{xcolor}
\definecolor{TUMblue}{HTML}{0065bd}
\usepackage[colorlinks=true,linkcolor=TUMblue,citecolor=OliveGreen,urlcolor=Sepia,linktocpage=true]{hyperref}
\usepackage{tgtermes}

\usepackage{amsmath}
\usepackage{amssymb}
\usepackage{nccmath}
\usepackage{pifont}
\usepackage{braket}
\usepackage{amsthm}
\usepackage{upgreek}
\usepackage{comment}
\usepackage{stmaryrd}
\usepackage{appendix}
\usepackage{bbm}
\usepackage{url}
\usepackage{tikz-cd}
\usepackage{caption}
\usepackage{subcaption}
\usepackage{dsfont}
\usepackage[shortlabels,inline]{enumitem}
\usepackage{graphbox}
\usepackage{xfrac}
\usepackage{collect}

\makeatletter
\def\moverlay{\mathpalette\mov@rlay}
\def\mov@rlay#1#2{\leavevmode\vtop{%
   \baselineskip\z@skip \lineskiplimit-\maxdimen
   \ialign{\hfil$\m@th#1##$\hfil\cr#2\crcr}}}
\newcommand{\charfusion}[3][\mathord]{
    #1{\ifx#1\mathop\vphantom{#2}\fi
        \mathpalette\mov@rlay{#2\cr#3}
      }
    \ifx#1\mathop\expandafter\displaylimits\fi}
\makeatother

\newtheorem{theorem}{Theorem}[section]
\newtheorem{corollary}[theorem]{Corollary}
\newtheorem{lemma}[theorem]{Lemma}
\newtheorem{proposition}[theorem]{Proposition}
\newtheorem{example}[theorem]{Example}
\newtheorem{definition}[theorem]{Definition}
\newtheorem{remark}[theorem]{Remark}

\newcommand{\generate}[2]{\langle #1 \rangle_{\mathsf{#2}}}

\newcommand{\iu}{\mathrm{i}\mkern1mu}

\renewcommand{\C}{\mathbb{C}}
\newcommand{\R}{\mathbb{R}}

\newcommand{\Z}{\mathbb{Z}}

\newcommand{\mf}[1]{\mathfrak{#1}}
\newcommand{\mc}[1]{\mathcal{#1}}

\newcommand{\sign}{\operatorname{sign}}
\newcommand{\diag}{\operatorname{diag}}

\newcommand{\id}{\mathds1}

\newcommand{\SU}{\operatorname{SU}}
\newcommand{\SO}{\operatorname{SO}}

\newcommand{\argmax}{\operatorname{argmax}}
\newcommand{\argmin}{\operatorname{argmin}}

\renewcommand{\Im}{\operatorname{Im}}
\renewcommand{\Re}{\operatorname{Re}}

\newcommand{\ad}{\operatorname{ad}}
\newcommand{\tr}{\operatorname{tr}}

\newcommand{\derv}{\mathsf{derv}}

\newcommand{\reachR}{\mathsf{reach_R}}
\newcommand{\reachB}{\mathsf{reach_B}}

\renewcommand{\epsilon}{\varepsilon}
\newcommand{\e}{\mathbf e}
\newcommand{\wkl}{{\mf w_{\sf{GKSL}}}}

\newcommand{\ketbra}[1]{\ket{#1}\!\bra{#1}}

\begin{document}

\title[Optimal Control of a Markovian Qubit with Unitary Control]{Optimal Control of a Markovian Qubit with Unitary Control}

\author{Emanuel Malvetti}
\address{School of Natural Sciences, Technische Universit\"at M\"unchen, 85737 Garching, Germany, and Munich Center for Quantum Science and Technology (MCQST) \& Munich Quantum Valley (MQV)}

\begin{abstract}
We study a single Markovian qubit governed by a Lindblad master equation and subject to fast unitary control.
Using reduced control systems and optimal control theory we determine (i) controls for cooling and heating such systems in a time-optimal way as well as (ii) the set of stabilizable states in the Bloch ball.
No restrictions on the Lindblad equation are assumed, and several known results, for instance for the Bloch equations, are recovered.
Furthermore we introduce integral systems, for which the solutions take a particularly nice form.
These integral systems include all systems with real Lindblad terms as well as all coolable systems.
The method allows for intuitive visualizations and is mostly analytical, making use of only basic numerical methods.

\bigskip

 \medskip

\end{abstract}

\maketitle

\section{Introduction}

Controlling individual qubits is a fundamental task in quantum information technologies, especially in the presence of noise and decoherence.
For instance cooling qubits to their ground state is essential for quantum computation~\cite{VincCriteria}.
In this paper we will use the method of reduced control systems~\cite{Reduced23,LindbladReduced23} and quantum control theory~\cite{DiHeGAMM08,dAless21} to derive time-optimal controls for a single qubit following Markovian time evolution and subject to fast unitary control.

Reduced control systems, under simplifying assumptions, have first been considered in~\cite[Ch.~22]{Agrachev04}, and these assumptions were recently shown to be unnecessary~\cite{Reduced23}.
In the context of Lindbladian systems with fast unitary control the reduced control system has been used in~\cite{Sklarz04,Yuan10,rooney2018} and~\cite{LindbladReduced23}.
The special case of the Bloch equations (corresponding exactly to the Lindblad equation with rotational invariance around the $z$-axis) was studied in~\cite{Lapert10,Lapert11,Lapert13}, introducing the so-called magic plane and steady state ellipsoid.
A special case of the Lindblad equations with bounded controls was studied using geometric methods in~\cite{Bonnard09,Bonnard09b}.
Another special case of the Bloch equations with incomplete control was addressed in~\cite{Lokutsievskiy24}.
The relaxation of certain unital channels and a special case of the Bloch equations was studied in~\cite{Mukherjee13} with bounded and unbounded controls.
The general Lindblad case was treated in~\cite{Rooney12,Rooney16,Rooney20} using mostly numerical methods.

Our approach uses a different method based on analyzing the generators of the reduced control system and yields a more comprehensive solution to the problem of finding time-optimal controls. 
The solution is almost completely analytical, except that as a final step one generally needs to use numerical integration, hinting at the fact that in the general case the solution does not admit an analytical expression.
Moreover our approach is very visual and geometric, giving an intuitive understanding of certain features of the obtained solutions.
A drawback of our approach is, however, that it does not easily extend to the case of bounded controls. 

\subsection*{Outline}

Section~\ref{sec:systems} introduces the setting and defines the relevant control systems.
Section~\ref{sec:generators} highlights the main tools and methods used in the paper and contains some preliminary results.
The main tasks addressed in the paper, namely optimal control and stabilization, are introduced and treated in a general fashion in Section~\ref{sec:tasks}.
The remainder of the main part focusses on deriving concrete solutions, starting with the completely general case in Section~\ref{sec:general-systems}.
The treatment of the general case leads to the definition of so-called integral systems for which the solution simplifies considerably.
These systems are addressed in Section~\ref{sec:integral}, and they include important special cases such as real and coolable systems.
The latter are considered in Section~\ref{sec:coolable}.
The study of some special systems (namely unital systems and the Bloch equations) is relegated to Appendix~\ref{sec:special-systems}.
Finally, Appendix~\ref{app:calc} contains some technical computations.

\section{Control systems} \label{sec:systems}

In this paper we consider an open two-level quantum system, henceforth called a qubit.
The uncontrolled Markovian evolution is described by the (time-independent) Lindblad equation~\cite{GKS76,Lindblad76}. 
Introducing fast unitary control to the system we obtain a bilinear control system~\cite{Jurdjevic97,Elliott09,DiHeGAMM08,dAless21}.
The assumption on the controls lead to a natural reduced control system as defined in~\cite{Reduced23,LindbladReduced23}.
This reduced control system describes the evolution of the eigenvalues of the density matrix representing the quantum state and will be studied in detail in the following sections.

\subsection{Bloch ball and Lindblad equation}

We start by recalling the basic formalism in more detail. 
The set of all possible mixed states of a qubit is given by the set of density matrices of size $2\times2$, which are exactly the positive semi-definite matrices of trace one, denoted
$\mf{pos}_1(2) = \{\rho\in \C^{2,2} : \rho\geq 0, \,\tr(\rho)=1\}$.
Using the \emph{Pauli matrices}
\begin{align*}
\sigma_x=\begin{pmatrix}
0&1\\1&0
\end{pmatrix},\quad
\sigma_y=\begin{pmatrix}
0&-\iu\\\iu&0
\end{pmatrix},\quad
\sigma_z=\begin{pmatrix}
1&0\\0&-1
\end{pmatrix},
\end{align*}
which form an orthonormal basis of the set $\iu\mf{su}(2)$ of traceless Hermitian matrices of size $2\times 2$ with respect to the (rescaled) Hilbert-Schmidt inner product $\braket{A,B}:=\tr(A^*B)/2$, we obtain the affine-linear isometry
$B(1/2)\to \mf{pos}_1(2)$ defined by $(x,y,z)\mapsto \tfrac{1}{2}\mathds1 + x\sigma_x + y\sigma_y + z\sigma_z$,
where $B(1/2)\subseteq\R^3$ denotes the ball of radius $1/2$, called the \emph{Bloch ball}\footnote{%
Sometimes the Bloch ball is defined such that is has radius $1$.}, cf.~\cite[Section~5.2]{Bengtsson17}.

The special unitary transformations $U\in\SU(2)$ act on the density matrices by conjugation $\rho\mapsto U\!\rho U^*$. 
Note that the kernel of this action is $\{\mathds1,-\mathds1\}$. 
In the Bloch ball picture, these transformations are rotations, belonging to the special orthogonal group $\SO(3)$.
The $\SU(2)$ orbits in $\mf{pos}_1(2)$ are exactly the sets of density matrices sharing the same eigenvalues, and they correspond to the concentric spheres of the Bloch ball. 
This illustrates the fact that the radius of a point in the Bloch ball only depends on the eigenvalues $\lambda,1-\lambda\in[0,1]$ of the corresponding density matrix, and in fact determines the eigenvalues up to their order. 
Indeed, if $\lambda\geq 1/2$ is the larger eigenvalue, the radius is given by $r=\lambda-1/2\in[0,1/2]$. 
Conversely, given radius $r$, the eigenvalues are $1/2\pm r$. 
We see that a state is pure, meaning $\rho=\ketbra{\psi}$, if and only if it corresponds to a point on the surface of the Bloch ball, and the interior of the Bloch ball consists of all mixed states. 
The center corresponds to the maximally mixed state $\mathds1/2$. 

To make this correspondence more precise we consider the isometric embedding
$\iota : [0,1]\to\mf{pos}_1(2)$ 
given by 
$\lambda \mapsto \tfrac{\mathds1}{2}+\left(\lambda-\tfrac{1}{2}\right)\sigma_z$
which maps $[0,1]$ to the subset of $\mf{pos}_1(2)$ consisting of diagonal density matrices corresponding exactly to the $z$-axis of the Bloch ball.
Note that if we endow $[0,1]$ with the metric induced by the absolute value, then $\iota$ is isometric since $\sqrt{\tr((\iota(\lambda)-\iota(\lambda'))^2)/2}=|\lambda-\lambda'|$.
A nice property of the set $\iota([0,1])$ is that it intersects all orbits orthogonally and the unitaries which leave $\iota([0,1])$ invariant act on it either trivially or by reflection about the origin. 
In $[0,1]$ this corresponds to the reflection $\lambda\mapsto1-\lambda$. 
To get rid of this final ambiguity one may work on the halved interval $[1/2,1]$.
Note that there is nothing special about the $z$-axis, except that it corresponds to the diagonal matrices in the standard basis. 
Any other axis would work for our purposes, since in fact all such axes are related by rotations of the Bloch ball.
When defining the reduced control system in the next section we will use $[0,1]$ as the reduced state space.

\smallskip
The \emph{\textsc{gks}--Lindblad equation}~\cite{GKS76,Lindblad76} for a qubit is given by
\begin{align*}
\dot\rho=-L(\rho)=-\iu[H_0,\rho] + \sum_{k=1}^r \big( V_k\rho V_k^* -\tfrac{1}{2}(V^*_k V_k\rho+\rho V^*_k V_k) \big),
\end{align*}
where the Hamiltonian $H_0\in\iu\mf{u}(2)$ is a Hermitian matrix and the \emph{Lindblad terms} $\{V_k\}_{k=1}^r\subset\C^{2,2}$ are arbitrary matrices. 
We call $-L$ a \emph{Lindblad generator}\footnote{The signs are chosen such that the real parts of the eigenvalues of $-L$ are non-positive.}, and we denote the set of all Lindblad generators in $2$ dimensions by $\wkl(2)$, called the \emph{Kossakowski--Lindblad Lie wedge}, cf.~\cite{DHKS08}.

\subsection{Full and reduced control systems}

In this paper we study a Markovian qubit subject to fast unitary control.
Let $I$ be an interval of the form $[0,T]$ or $[0,\infty)$.
We say that $\rho:I\to\mf{pos}_1(2)$ is a solution to the bilinear control system~\cite{Jurdjevic97,Elliott09,DiHeGAMM08}
\begin{align} \label{eq:bilin} \tag{$\mathsf F$}
\dot\rho(t) = -\Big(\iu\sum_{j=1}^m u_j(t)[H_j,\rho(t)] + L(\rho(t))\Big), \quad\rho(0)=\rho_0\in\mf{pos}_1(2), 
\end{align}
if it is absolutely continuous and satisfies the equation almost everywhere.
Here $-L\in\wkl(2)$ denotes the \emph{drift Lindblad generator} (with \emph{Lindblad terms} $V_k$ and including a possible Hamiltonian part $H_0$) describing the uncontrolled evolution of the system, the $H_j\in\iu\mf{u}(2)$ for $j=1,\ldots,m$ denote the \emph{control Hamiltonians}, and the functions $u_j:I\to\R$ are the \emph{control functions}.
Throughout we make the following two crucial assumptions: 
First we only require the control functions to be locally integrable, meaning that we do not assume any bounds. 
Second we assume that the control Hamiltonians generate the full special unitary Lie algebra $\generate{\iu H_j:j=1,\ldots,m}{\mathsf{Lie}}\supseteq\mf{su}(2)$.
Taken together this means that we have fast unitary control over the system.

\smallskip
Under these assumptions an equivalent reduced control system can be defined, cf.~\cite[Sec.~2.2]{LindbladReduced23}, by focussing on the evolution of the eigenvalues of $\rho$.
More precisely the reduced state will be $\lambda\in[0,1]$.
First we define the matrices $J_{ij}(U) = \sum_{k=1}^r|(U^* V_k U)_{ij}|^2$.
For each $U\in\SU(2)$, we obtain an \emph{induced vector field} on $[0,1]$ defined by
\begin{align} \label{eq:line}
\lambda \mapsto -Q_U(\lambda) = J_{12}(U)-\lambda(J_{12}(U)+J_{21}(U)).
\end{align}
Concretely each $-Q_U$ is an affine linear function on $[0,1]$.
This allows us to define the set-valued function
\begin{align*}
\derv:[0,1]\to\mc P(\R),\qquad \derv(\lambda) = \{-Q_U(\lambda) : U\in\SU(2)\},
\end{align*}
of \emph{achievable derivatives} (where $\mc P(\cdot)$ denotes the power set).
For an example of $\derv$ see Figure~\ref{fig:example-space-of-lines}. Then the \emph{reduced control system} on $[0,1]$ is defined by the differential inclusion~\cite{Smirnov02,Aubin84}%
\footnote{There are some slightly different ways to define the reduced control system, see~\cite[Sec.~2.2]{LindbladReduced23}, but the distinction is not relevant for us.}
\begin{align} \label{eq:reduced} \tag{$\mathsf R$}
\dot\lambda(t)\in\derv(\lambda(t)),\quad \lambda(0)=\lambda_0\in[0,1]. 
\end{align}

The Equivalence Theorem, see~\cite[Thm.~2.6]{LindbladReduced23}, shows that under the present assumption of fast unitary control, the bilinear control system~\eqref{eq:bilin} is equivalent to the reduced control system~\eqref{eq:reduced} in a precise sense.
Importantly, no loss of information is incurred by switching to the reduced control system.
Essentially we reduced the state space of the control system to the $z$-axis of the Bloch ball (representing the radius of the state), and we only consider the movement of the state along this axis.
As a consequence we obtain the following equivalence of reachable sets, cf.~\cite[Prop.~4.1]{LindbladReduced23}.
First recall that the \emph{reachable set} of $\lambda_0$ at time $T\geq0$ of the reduced control system~\eqref{eq:reduced}, 
denoted $\reachR(\lambda_0,T)$, is the set of all $\lambda(T)$ where $\lambda:[0,T]\to[0,1]$ is a solution to~\eqref{eq:reduced} with $\lambda(0)=\lambda_0$. 
The (all-time) reachable set is $\reachR(\lambda_0)=\bigcup_{T\geq0}\reachR(\lambda_0,T)$.
The definitions for other control systems are entirely analogous.

\begin{proposition}
Let $\rho_0\in\mf{pos}_1(2)$ have eigenvalues $\{\lambda_0,1-\lambda_0\}$, then for all $T>0$ it holds that
$$
\overline{\reachB(\rho_0,T)} = \overline{\{U\iota(\lambda)U^* : \lambda\in\reachR(\lambda_0,T), U\in\SU(2)\}},
$$
and similarly for the (all-time) reachable sets.
\end{proposition}

\begin{remark}
In~\cite{LindbladReduced23} the reduced control system is defined on the standard simplex $\Delta^{n-1}$ in $\R^n$.
Using the embedding $[0,1]\to\Delta^1\subset\R^2$, given by $\lambda\mapsto(\lambda,1-\lambda)^\top$ we pulled back the control system to the interval $[0,1]$.
This turns the (stochastic) linear dynamics on $\R^2$ into affine linear dynamics on $[0,1]$.
\end{remark}

\section{Space of generators and optimal derivatives} \label{sec:generators}

The induced vector fields $-Q_U$ generate the dynamics of the reduced control system~\eqref{eq:reduced}.
Due to the present low-dimensional setting, it turns out that these generators sit in a two-dimensional vector space and hence they can easily be visualized.
Understanding the exact shape of the set of generators is non-trivial, but can be done analytically, and this ultimately leads to solutions for the optimal control problem of~\eqref{eq:bilin}.
Furthermore it yields a parametrization of the stabilizable states in the Bloch ball.

Since each $\derv(\lambda)$ is the image of a continuous function on the compact connected set $\SU(2)$, the set $\derv(\lambda)$ must be a closed bounded interval in $\R$. 
Hence we can define the \emph{optimal derivative function}
\begin{align*}
\mu:[0,1]\to\R,\quad \lambda\mapsto\max\derv(\lambda).
\end{align*}
To fully understand the graph associated to $\derv$ (when seen as a set-valued function) it suffices to study the function $\mu$, since $\lambda\mapsto-\mu(1-\lambda)$ is the corresponding lower boundary due to the reflection symmetry on $[0,1]$ shown in Lemma~\ref{lemma:sol-basic} below.
The optimal derivative function $\mu$ enjoys some nice properties:

\begin{lemma} \label{lemma:mu-basic}
The function $\mu:[0,1]\to\R$ is continuous, convex and non-increasing.
Furthermore $\mu(1/2)$ is equal to the larger (non-negative) eigenvalue of $\sum_{k=1}^r [V_k,V_k^*]/2$. 
In particular, $\mu$ is non-negative on $[0,1/2]$ and non-positive at $1$.
\end{lemma}

\begin{proof}
By definition, the function $\mu$ can be seen as the pointwise maximum of the decreasing affine linear functions $-Q_U$ parametrized by $U\in\SU(2)$.
From this it follows that $\mu$ is continuous, convex, and non-increasing. 
The last fact follows from
$\mu(1/2) 
=
\max_U\tfrac{1}{2}(J_{12}(U)-J_{21}(U))
=
\max_U\tfrac{1}{2}((J(U)-J(U)^T)\e)_1
=
\max_U\tfrac{1}{2} \sum_{k=1}^r(U^*[V_k,V_k^*]U)_{11}$,
where the maximization is over $\SU(2)$ and where $\e=(1,1)/2$ and we used~\cite[Lem.~B.2]{LindbladReduced23}.
\end{proof}

Each induced vector field $-Q_U$ is defined by the values taken at the endpoints, namely $J_{12}(U)\geq0$ at $\lambda=0$ and $-J_{21}(U)\leq0$ at $\lambda=1$. 
This motivates the definition of the \emph{space of generators} as
\begin{align*}
\mf Q = \{(J_{12}(U)-J_{21}(U),J_{12}(U)+J_{21}(U)) : U\in\SU(2)\} \subset\R^2.
\end{align*}
This set is clearly linearly isomorphic to the set $\{-Q_U: U\in\SU(2)\}$ but has the advantage of being easy to visualize.
For an example see Figure~\ref{fig:example-space-of-lines}. 
Understanding the space of generators $\mf Q$, and more specifically its boundary, allows us to describe the function $\mu$, which in turn allows us to find solutions to the optimal control problem.

\begin{lemma} \label{lemma:sol-basic}
The set $\mf Q$ is compact, path-connected, and satisfies $y\geq|x|$ for all $(x,y)\in\mf Q$. 
Moreover $\mf Q$ is symmetric with respect to the reflection $(x,y)\mapsto(-x,y)$.
\end{lemma}

\begin{proof}
The set $\mf Q$ is the image of a continuous function on $\SU(2)$, and hence compact and path-connected. 
By definition $J_{12}(U),J_{21}(U)\geq0$, and so $J_{12}(U)+J_{21}(U)\geq|J_{12}(U)-J_{21}(U)|$. 
For any $U$ it holds that $J_{12}(U\sigma_x)=J_{21}(U)$ and $J_{21}(U\sigma_x)=J_{12}(U)$, proving the symmetry.
\end{proof}

\begin{figure}[thb]
\centering
\includegraphics[width=0.45\textwidth]{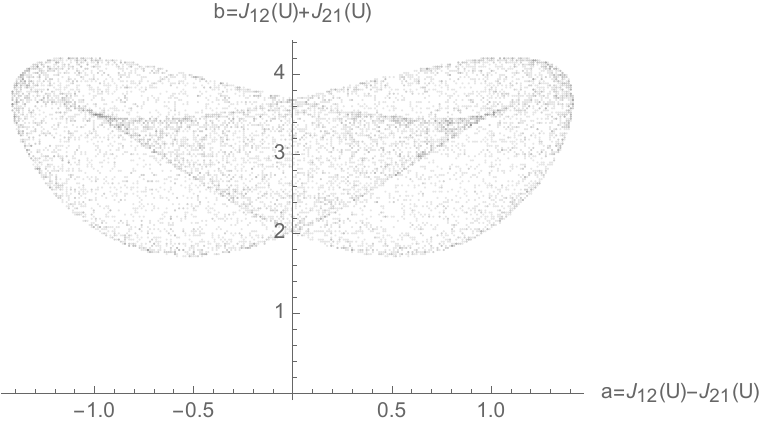}
\includegraphics[width=0.45\textwidth]{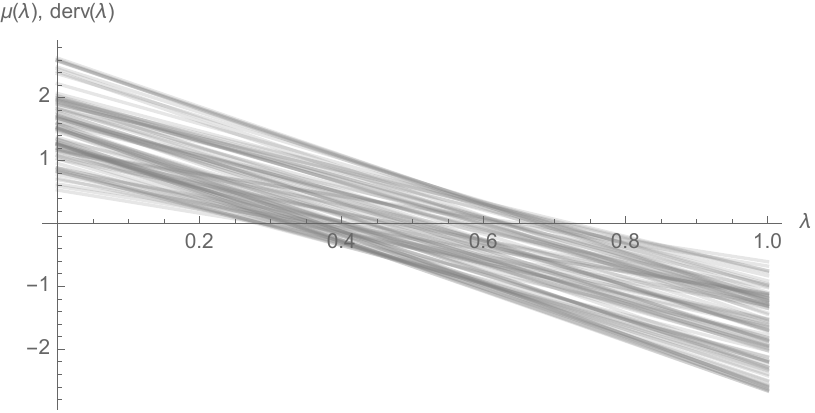}
\caption{We consider a generic Lindblad generator $-L\in\wkl$.
Left: The space of generators $\mf Q$ plotted using randomly sampled points.
Clearly $\mf Q$ has some intriguing structure, and hence it is the first object we wish to understand.
Right: Graph of the set-valued function $\derv$ generated by plotting the affine linear functions $-Q_U$ associated to randomly selected unitaries $U$. The functions $\mu(\lambda)$ and $-\mu(1-\lambda)$ are the upper and lower boundaries of the set of all such lines. 
The Lindblad terms $V_k$ are
$\big(\begin{smallmatrix} 0.0-0.9\iu & -0.6+0.6\iu \\ 0.+0.8\iu & 0.9-\iu \end{smallmatrix}\big),
\big(\begin{smallmatrix} -0.1+0.8\iu & 0.3-0.3\iu \\ -0.8+0.6\iu & 0.3+0.1\iu \end{smallmatrix}\big)$, and
$\big(\begin{smallmatrix} -0.2+0.4\iu & 0.6-0.2\iu \\ 0.6-0.7\iu & 0.2+0.8\iu \end{smallmatrix}\big)$.
}
\label{fig:example-space-of-lines}
\end{figure}

When $(0,0)\in\mf Q$ the system is of a special type (called unital stabilizable) which we will explore in Appendix~\ref{app:unital}.
In the main part we mostly focus on the case $(0,0)\notin\mf Q$.
A first relation between the objects $\mf Q$ and $\mu$ is given by the following result:

\begin{lemma} \label{lemma:mu-strict}
If $(0,0)\notin\mf Q$ then $\mu$ is strictly decreasing.
\end{lemma}

\begin{proof}
By contraposition, if $\mu$ is not strictly decreasing, there are two points, $\lambda_1<\lambda_2$ in $[0,1]$ such that $\mu(\lambda_1)=\mu(\lambda_2)$ and $\mu$ will be constant on $[\lambda_1,\lambda_2]$. Then there must be some horizontal line of the form~\eqref{eq:line} passing through $(\lambda,\mu(\lambda))$ for $\lambda\in(\lambda_1,\lambda_2)$. Let $U\in\SU(2)$ be a corresponding unitary. Then the slope of the line is $-(J_{12}(U)+J_{21}(U))$ and equals $0$, and so $J_{12}(U)=J_{21}(U)=0$, as desired.
\end{proof}
\noindent The space of generators $\mf Q$ and the optimal derivative function $\mu$ are linked by the following key result:

\begin{proposition} \label{prop:polarity}
Let $\lambda\in[0,1]$. 
Then there is some line of the form~\eqref{eq:line} passing through $(\lambda,\mu(\lambda))$, and a point in $\mf Q$ corresponds to such a line if and only if it solves the linear optimization problem
$$
\mu(\lambda)=\max_{(a,b)\in\mf Q}\tfrac{1}{2} (a+(1-2\lambda)b).
$$
\end{proposition}

\begin{proof}
The optimization problem follows immediately from the definition of $\mu$. 
The existence of a solution follows from the fact that $\mf Q$ is compact.
\end{proof}

\begin{remark}
Note that the relation between the function $\mu:[0,1]\to\R$ and the set $\mf Q\subseteq\R^2$ is similar to the Legendre--Fenchel transform of a function and the polar dual of a polytope.
In particular corners of $\mf Q$ yield affine linear parts of $\mu$ as can be seen for instance in Appendix~\ref{subsec:bloch-eq}.
Note also that $\mf Q$ need not be convex, as can be seen in Figure~\ref{fig:example-space-of-lines}, and $\mu$ only depends on the convex hull of $\mf Q$.
\end{remark}

The space of generators $\mf Q$ fully describes the reduced control system~\eqref{eq:reduced}, and thus allows us to compute for instance reachable and stabilizable spectra. 
In order to compute the fastest path in the Bloch ball, along with the optimal controls, as well as the stabilizable states in the Bloch ball, we need to parametrize $\mf Q$ in terms of the corresponding unitaries. 
More precisely, it suffices to consider a subset of $\SU(2)$ which can map the north pole of the Bloch sphere to any other point on the Bloch sphere. 
We will give this parametrization in full generality in Section~\ref{sec:general-systems}, and show how it simplifies in special cases of interest. Using the unitaries $U_{x,z}=\exp(\iu\pi z \sigma_z)\exp(\iu\pi x \sigma_x)$, the parametrization will be of the form
\begin{align} \label{eq:sol-param}
(x,z)\mapsto (J_{12}(U_{x,z})-J_{21}(U_{x,z}),J_{12}(U_{x,z})+J_{21}(U_{x,z}))
\end{align}
for $x\in [0,1/2]$ and $z\in[0,1)$. Hence $2\pi x$ corresponds to the polar angle and $2\pi z$ to the azimuthal angle on the Bloch sphere.

\section{Control tasks} \label{sec:tasks}

The optimal derivative function $\mu:[0,1]\to\R$ discussed in the previous section gives the fastest increase (or slowest decrease if it is negative) of $\lambda$ which can be achieved in the reduced control system.
It allows us to directly read off relevant information such as the reachable and stabilizable sets from its graph.
Knowledge of the function $\mu$ also allows us to determine optimal controls for the full control system, as we will show in this section.

\subsection{Stabilizable set}

A relevant task in applications is that of stabilizing the system in a certain desired state.
We will explicitly determine which states are stabilizable in the full and in the reduced control system.
Moreover we will show concretely how such states can be stabilized.

We begin with the reduced control system. 
For $\lambda\in[0,1]$ we say that $\lambda$ is \emph{stabilizable}\footnote{In~\cite{Reduced23,LindbladReduced23} we distinguish between stabilizable and strongly stabilizable states, but in the present setting the two notions coincide.} if $0\in\derv(\lambda)$.
Concretely this means that in the reduced control system, the constant path at $\lambda$ is a solution.

\begin{lemma}
Assume that $(0,0)\notin\mf Q$. 
Then the set of stabilizable states is the non-empty closed interval $[1-\lambda^\star,\,\lambda^\star]$ where $\lambda^\star\geq1/2$ is the unique root of $\mu$.
If $(0,0)\in\mf Q$, then all states $\lambda\in[0,1]$ are stabilizable.
\end{lemma}

\begin{proof}
The case of $(0,0)\in\mf Q$ is clear, so assume that $(0,0)\notin\mf Q$.
By Lemma~\ref{lemma:mu-strict} it holds that $\mu$ has at most one root. 
By Lemma~\ref{lemma:mu-basic} there must be at least one root in the interval $[1/2,1]$. Denote this root by $\lambda^\star$. The lower bound of $\derv$ is given by $\lambda\mapsto-\mu(1-\lambda)$ which has a unique root at $1-\lambda^\star$. Since both $\mu$ and the corresponding lower boundary are strictly decreasing, the stabilizable region is $[1-\lambda^\star,\,\lambda^\star]$.
\end{proof}
\noindent We call $\lambda^\star$ the \emph{purest stabilizable state}. 
The stabilizable set can be obtained graphically from the space of generators as follows.

\begin{lemma} \label{lemma:stab}
Assume that $(0,0)\notin\mf Q$.
It holds that $\lambda\in[0,1]$ is stabilizable if and only if $\lambda=\frac12(1+\frac{a}{b})$ for some $(a,b)\in\mf Q$, and hence $\lambda^\star=\max_{(a,b)\in\mf Q}\frac12(1+\frac{a}{b})$.
\end{lemma}

\begin{proof}
Since $(0,0)\notin\mf Q$ the value $\frac{a}{b}$ is well defined for every point $(a,b)\in\mf Q$, and every line of the form~\eqref{eq:line} intersects the abscissa in a unique point, namely $\lambda=\frac12(1+\frac{a}{b})$.
By Lemma~\ref{lemma:sol-basic} the value $\frac{a}{b}$ is always contained in $[-1,1]$, and the set of all possible $\frac{a}{b}$ is a non-empty closed interval symmetric about $0$. In particular is has a greatest element.
\end{proof}
\noindent In practice for general systems, after parametrizing the set $\mf Q$ as discussed in the previous section, it is necessary to use a root-finding algorithm (such as the bisection method) to find $\lambda^\star$, cf.~\cite{Epperson13}.

The following result presents a class of systems for which the purest stabilizable state $\lambda^\star$ can be found analytically.

\begin{lemma} \label{lemma:heat-baths}
Let $V$ be an arbitrary Lindblad term and consider the system defined by the Lindblad terms $V$ and $\sqrt{\gamma} V^*$ with $\gamma\in[0,1]$. 
Then the purest stabilizable state is $\lambda^\star=\frac{1}{1+\gamma}$, unless $V$ is normal, in which case it is $\lambda^\star=1$.
\end{lemma}

\begin{proof}
By Lemma~\ref{lemma:unital-stab}~\ref{it:unit-stab-normal}, if $V$ is normal, then the system is unital stabilizable, and so the stabilizable region is $[0,1]$. 
Now assume that $V$ is not normal, then, by the same lemma, $(0,0)\notin\mf Q$. 
First we consider the case $\gamma=0$. 
Then, as shown in~\cite[Sec.~V]{OptimalCooling24}, the system is coolable, and again the stabilizable region is $[0,1]$. 
Now let $\gamma\in(0,1]$.
Note that for any $U\in\SU(2)$, if we denote $t(U)=\frac{J_{21}(U)}{J_{12}(U)}$, the intersection of the corresponding line with the abscissa is $\frac{J_{12}(U)}{J_{12}(U)+J_{21}(U)} = \frac{1}{1+t(U)}$ by Lemma~\ref{lemma:stab}.  
Hence it suffices to find the minimal value of the ratio $t(U)$. 
If we denote by $J'(U)$ the matrix corresponding to $V$ then
$t(U)=\frac{J_{21}(U)}{J_{12}(U)}=\frac{J_{21}'(U)+\gamma J_{12}'(U)}{J_{12}'(U) + \gamma J_{21}'(U)}$.
Since all quantities are non-negative, and $\gamma\leq1$, it follows from the mediant inequality that $\gamma\leq t(U)\leq\tfrac1\gamma$.
Moreover $t(U)=\gamma$ is achieved when $V$ is in upper triangular form, which is always possible. 
\end{proof}
\noindent Note that when $\gamma=1$ the system is unital (cf.~Appendix~\ref{app:unital}), and if $\gamma=0$ we obtain a rank one system, see~\cite[Sec.~V]{OptimalCooling24}.

\smallskip
In Lemma~\ref{lemma:stab} we have described the stabilizable spectra in the reduced control system, that is, in $[0,1]$. 
However we can also explicitly describe the set of stabilizable density matrices. 
In~\cite{Lapert13} this was done for the special case of the Bloch equations obtaining an ellipsoid, which we will also recover in Section~\ref{subsec:bloch-eq}. 

A state $\rho\in\mf{pos}_1(2)$ is called \emph{stabilizable} if there exists a control Hamiltonian $H_c$ turning $\rho$ into a fixed point, that is, if $-L(\rho)-\iu[H_c,\rho]=0$. 
The following result connects the two notions of stability and is a restatement of~\cite[Prop.~3.2]{LindbladReduced23}.

\begin{proposition} \label{prop:stab}
Assume that $-L$ is a non-unital Lindblad generator and let $\lambda\neq\tfrac12$.
Then the following are equivalent:
\begin{enumerate}[(i)]
\item There is some $U$ such that $-Q_U(\lambda)=0$, that is, $\lambda$ is stabilizable for the reduced control system~\eqref{eq:reduced}.
\item There is some $U$ and $H_c$ such that for $\rho=U\iota(\lambda)U^*$ we have $-L(\rho)-\iu[H_c,\rho]=0$, that is, $\rho$ is stabilizable for the full control system~\eqref{eq:bilin}.
\end{enumerate}
Moreover any unitary satisfying one of the above also satisfies the other, and given $\lambda$ and $U$ one can compute a corresponding \emph{compensating Hamiltonian} $\iu H_c=\ad_{\rho}^+(L(\rho))$, where $\ad_\rho(\cdot)=[\rho,\cdot]$ and $(\cdot)^+$ denotes the Moore--Penrose pseudoinverse.
\end{proposition}

\begin{remark}
Whenever $\lambda=1/2$, it is clear that $\lambda$ is stabilizable in the reduced system, but in general it is not (exactly) stabilizable in the full system.
Indeed, in this case the formula for the compensating Hamiltonian might diverge.
The only exception is the case of unital $-L$, where the maximally mixed state is always a fixed point independent of the applied controls.
\end{remark}

The set of stabilizable states in the Bloch ball can be nicely parametrized using a parametrization of $\mf Q$.
Assume that $(0,0)\notin\mf Q$ and let $(x,z)\mapsto F(x,z)$ be the parametrization of $\mf Q$ as in~\eqref{eq:sol-param}. 
By construction, the line represented by $F(x,z)$ corresponds to the affine linear vector field obtained from $-L$ when restricting to and projecting onto the axis passing through the point with polar angle $\theta=2\pi x$ and azimuthal angle $\phi=2\pi z$. 
A stabilizable point on the axis is exactly a zero of this vector field.
Together with Lemma~\ref{lemma:stab} this shows the following result:

\begin{lemma}  \label{lemma:stab-ball}
The stabilizable set can be parametrized as 
$r(\theta,\phi)=\frac12 \frac{J_{12}(U_{x,z})-J_{21}(U_{x,z})}{J_{12}(U_{x,z})+J_{21}(U_{x,z})}$
in spherical coordinates $(r,\theta,\phi)$, where $2\pi x=\theta$ and $2\pi z=\phi$.
\end{lemma}

In the Bloch ball, the shape of the set of stabilizable states is some kind of ovoid, and in some special cases it is an ellipsoid.
Moreover, the intersection of this set with any plane containing the $z$-axis (after an appropriate change of basis) is in fact an ellipse, cf.~Proposition~\ref{prop:potato}.

\subsection{Optimal controls} \label{sec:optimal-controls}

The main task of interest is to determine which states are reachable from a given initial state, and to compute the fastest path to reach such a state, together with the corresponding controls.
The Equivalence Theorem, see~\cite[Thm.~2.6]{LindbladReduced23}, allows us to work on the level of the reduced control system~\eqref{eq:reduced} and then to lift the obtained result to the full control system~\eqref{eq:bilin}.

Indeed, since the reduced state space $[0,1]$ is one-dimensional, the shortest path is always uniquely defined up the speed at which the path is traversed.
Clearly the maximal (positive) velocity a solution can achieve at $\lambda$ is the optimal derivative $\mu(\lambda)$, and hence the optimal solution $\lambda(t)$ is obtained by integrating along $\mu$. 
Since there are only two directions in the one-dimensional reduced state space, there are only two tasks to consider, namely optimal heating (mixing) and optimal cooling (purifying).
These problems can be solved together by determining the fastest path from $0$ to the purest stabilizable state $\lambda^\star$ in the reduced state space $[0,1]$.

In this section we outline the general approach we will take to solve this problem, and the concrete results will be presented in the subsequent sections.
As stated at the end of Section~\ref{sec:generators}, we will parametrize the space of generators $\mf Q$ using the unitaries $U_{x,z}=\exp(\iu\pi z \sigma_z)\exp(\iu\pi x \sigma_x)$ where the parameter ranges are $x\in[0,1]$ and $z\in[0,1)$.
Then, due to Proposition~\ref{prop:polarity} we are mainly interested in the boundary of $\mf Q$.
In general we will parametrize the boundary using the curves
\begin{align*}
J^+\to\R^2, \,\alpha\mapsto(a^+(\alpha), b^+(\alpha)) \qquad
J^-\to\R^2, \,\alpha\mapsto(a^-(\alpha), b^-(\alpha))
\end{align*}
defined on some intervals $J^+$ and $J^-$ for the upper (corresponding to heating) and lower (corresponding to cooling) part of the boundary respectively.
It is important to also determine the unitaries, or more precisely the values of the parameters $x,z$, which achieve the boundary points of $\mf Q$. 
For this we will determine parametrizations $x^\pm:J^\pm\to\R^2$ and $z^\pm:J^\pm\to\R^2$.

\begin{lemma} \label{lemma:analytic-params}
Given parametrizations $a^\pm, b^\pm, x^\pm, z^\pm$ as above, and assuming that $(a^\pm,b^\pm)$ is differentiable with non-zero derivative, one can find parametrizations
$\lambda^\pm(\alpha)$, $\mu^\pm(\alpha)$ and $\rho^\pm(\alpha)$ such that $(\lambda^\pm,\mu^\pm)$ parametrizes the graph of $\mu:[0,1]\to\R$ and such that $\rho^\pm(\alpha)$ parametrizes the optimal path through the Bloch ball.
Indeed we have the following expressions:
\begin{gather*}
\lambda^\pm(\alpha)=\tfrac12\big(1+\tfrac{a'(\alpha)}{b'(\alpha)}\big), \quad
\mu^\pm(\alpha)=\tfrac12(a^\pm(\alpha)+b^\pm(\alpha))-b^\pm(\alpha)\lambda^\pm(\alpha), \\
\rho^\pm(\alpha)=
U_{x^\pm(\alpha),z^\pm(\alpha)}\iota(\lambda^\pm(\alpha))U^*_{x^\pm(\alpha),z^\pm(\alpha)}.
\end{gather*}
\end{lemma}

\begin{proof}
By Proposition~\ref{prop:polarity} the value $\mu(\lambda)$ is obtained by maximizing a linear functional over $\mf Q$.
This maximum must be achieved on the boundary and hence there is some $\alpha$ such that $(a^\pm(\alpha),b^\pm(\alpha))$ achieves the maximum.
In this case it holds that the derivative of the parametrization is orthogonal to the direction of maximization, which immediately yields the expression for $\lambda^\pm(\alpha)$ and the expression for $\mu^\pm(\alpha)$ follows at once.
Then, by construction $\rho^\pm(\alpha)$ achieves the maximal eigenvalue derivative.
\end{proof}

The parametrizations of the previous lemma are completely analytic and allow to solve a significant part of the general problem.
The final goal is to determine the corresponding control functions for the full control system~\eqref{eq:bilin}.
Since these are functions of time all of the above quantities must also be expressed as a function of time.
To find the time parametrization $\alpha(t)$ one has to solve the following ordinary differential equation, where we omit the $\pm$ superscript for readability.

\begin{align} \label{eq:real-time-ode}
\alpha' = \frac{\mu(\alpha)}{\lambda'(\alpha)} =
\frac{a(\alpha)b'(\alpha)^2-a'(\alpha)b'(\alpha)b(\alpha)}
{b'(\alpha)a''(\alpha)-a'(\alpha)b''(\alpha)}.
\end{align}

Unfortunately it seems that in general the real time parametrization $\alpha(t)$ cannot be found analytically since the differential equation~\eqref{eq:real-time-ode} is too complicated.
Notable exceptions to this are however the Bloch equations (cf.~Section~\ref{subsec:bloch-eq}) and rank one systems~\cite[Sec.~V]{OptimalCooling24}.
For this reason numerical methods are indispensible.
Indeed for computational efficiency it can be beneficial to switch to a numerical representation of the functions right away instead of working with the analytical expressions which tend to become extremely convoluted.

Finally the corresponding control Hamiltonian can be found using the following result which is a simplified special case of~\cite[Prop.~3.10]{Reduced23}, but it can also easily be verified via direct computation. See also~\cite[Lem.~3.1]{Rooney12}.

\begin{proposition} \label{prop:lift}
Let $\rho:[0,T]\to\mf{pos}_1(2)\setminus\{\id/2\}$ be an achievable path of density matrices not passing through the maximally mixed state and let $\rho(t)=U(t)\iota(\lambda(t))U^*(t)$ be an eigenvalue decomposition of $\rho$. 
Then $\iu H(t) = -U'(t)U^*(t) + \ad_{\rho(t)}^+(L(\rho(t))$
is a path of skew-Hermitian matrices satisfying
$\dot\rho=-(\iu\ad_H+L)(\rho)$.
\end{proposition}
\noindent Note that the second term in the definition is analogous to the definition of the compensating Hamiltonian $H_c$ in Proposition~\ref{prop:stab}.

For our optimal control task the solution will cross the maximally mixed state $\{\id/2\}$ at one point in time.
The problem that occurs in trying to apply Proposition~\ref{prop:lift} is that the compensation term might blow up.
Luckily, as we will see in the following sections, this does not happen and the control functions tend to be very well behaved.
Nonetheless diverging controls can and do occur in certain special cases, such as the Bloch equations and rank one systems mentioned above.
In these cases the optimal solution is not differentiable at one point, where it takes a sharp turn, and the direct term diverges.
Still one can cut off this divergence at the price of an arbitrarily small error.

\section{General systems} \label{sec:general-systems}

The goal of this section is to implement the program set out in the previous sections.
We consider a qubit system described by an arbitrary finite set of Lindblad terms $V_k$ for $k=1,\ldots,r$. 
The main result is an analytical parametrization of the space of generators $\mf Q$, see Figure~\ref{fig:general-lines-points} for an example.
As a consequence we can determine the stabilizable states in the Bloch ball and the optimal controls in the original control system, cf.\ Figure~\ref{fig:general-opt-path}.
We will assume that $-L$ is not unital, since this special case is considered in Appendix~\ref{app:unital}.

\subsection{Parametrization} \label{sec:gen-param}

Let $\{V_k\}_{k=1}^r$ be a finite set of Lindblad terms. 
Without loss of generality we assume that all $V_k$ are traceless and that $\sum_{k=1}^r [V_k,V_k^*]$ is diagonal. 
We define the characteristic values
\begin{align*}
\Delta = J_{12}(\mathds1)-J_{21}(\mathds1), \quad
\Sigma = J_{12}(\mathds1)+J_{21}(\mathds1), \quad
\delta = 2J_{11}(\mathds1)-\Sigma/2,
\end{align*}
Furthermore we set
\begin{align*}
r_1 e^{i\phi_1} = \iu \sum_{k=1}^r \overline{(V_k)_{12}} (V_k)_{21},
\quad r_2 e^{i\phi_2} 
=4\sum_{k=1}^r \overline{(V_k)_{11}}(V_k)_{21},
\quad\phi&=\phi_1-2\phi_2+\pi/2.
\end{align*}
By choosing the basis appropriately we can always make sure that additionally $\Delta\geq0$, and in the following we will always assume that this is the case.\footnote{Since we assume that the system is non-unital it even holds that $\Delta>0$.}
The first crucial property is that these values are actually well-defined.

\begin{lemma} \label{lemma:well-def}
Assume that $-L\in\wkl$ is non-unital.
Then the values $|\Delta|$, $\Sigma$, $\delta$, $r_1$, and $r_2$ are well-defined, and if $r_1$ and $r_2$ are non-zero, then $\phi\mod2\pi$ is also well-defined.
More precisely this means that these values only depend on the generator $-L$, and not on the choice of $V_k$ or on the basis which diagonalizes $\sum_{k=1}^r [V_k,V_k^*]$.
\end{lemma}

\begin{proof}
We already assumed that the $V_k$ are traceless. This can always be done at the expense of a constant compensating Hamiltonian, hence we only have to show that the constants are invariant under unitary reshuffling, cf.~\cite[Lem.~C.3]{LindbladReduced23} or a change of basis keeping $\sum_{k=1}^r [V_k,V_k^*]$ diagonal.
It is easy to see that all quantities are invariant under unitary reshuffling, and the only allowed unitary basis transformations are those induced by diagonal unitaries, which can change $\phi_1$ and $\phi_2$ but not $\phi$.
\end{proof}
\noindent Using these characteristic values we can now give a parametrization of the space of generators $\mf Q$:

\begin{proposition}
\label{prop:general-parametrization}
The unitary $U_{x,z}=\exp(\iu\pi z\sigma_z)\exp(\iu\pi x\sigma_x)$ for $x\in[0,1/2]$, and $z\in[0,1)$ yields the point
$(J_{12}(U_{x,z})-J_{21}(U_{x,z}),J_{12}(U_{x,z})+J_{21}(U_{x,z})) 
= 
(\Delta\cos(2\pi x), F(x,z)),$
where
\begin{align*}
F(x,z) 
= 
\Sigma + (\delta + r_1\sin(4\pi z+\phi_1))\sin(2\pi x)^2
-r_2 \cos(2\pi x)\sin(2\pi x)\sin(2\pi z+\phi_2)
\,, 
\end{align*}
and hence we obtain the following parametrization of the space of generators
\begin{align*}
\mf Q=\{(\Delta\cos(2\pi x), \,F(x,z)):x\in[0,1/2], \,z\in[0,1)\}\,.
\end{align*}
\end{proposition}

\begin{proof}
The details of this elementary but lengthy computation are given in Appendix~\ref{app:parametrization}. 
\end{proof}
\noindent Note that for fixed $z$, the parametrization can be seen as the graph of a function. 
See Figure~\ref{fig:general-lines-points} for an illustration.
The parametrization of $\mf Q$ leads directly to a parametrization of the stabilizable states in the Bloch ball as shown in Lemma~\ref{lemma:stab-ball}, see Figure~\ref{fig:general-opt-path} for an example.

\begin{proposition} \label{prop:potato}
The parametrization of the set of stabilizable states in the Bloch ball is given in spherical coordinates by
$$
r(\theta,\phi) 
= \frac{\Delta\cos(\theta)}{2F(\tfrac{\theta}{2\pi},\tfrac{\phi}{2\pi})}
$$
For fixed angle $\phi$ this is an ellipse in the upper halfplane going through the origin.
\end{proposition}

\begin{proof}
The parametrization follows immediately from Lemma~\ref{lemma:stab-ball} and it is easy to see that for fixed angle $\phi$ the formula takes the form given in Lemma~\ref{lemma:ellipse-eq} (due to the definition of the polar angle $\theta$ all occurences of $\sin$ and $\cos$ in these formulas are swapped).
Note also that the shape of the ellipse can then be computed using Lemma~\ref{lemma:ellipse-params}.
\end{proof}

Even though the stabilizable set looks somewhat like an ellipsoid, its shape is more complicated.
Nevertheless, in some special cases, such as the Bloch equations treated in~\cite{Lapert13}, it indeed reduces to an ellipsoid, cf.\ Appendix~\ref{subsec:bloch-eq}.

\begin{figure}[thb]
\centering
\includegraphics[width=0.45\textwidth]{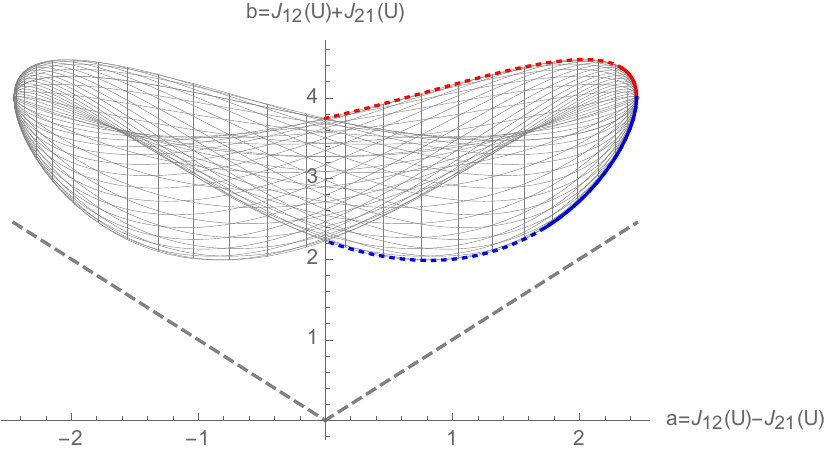}
\includegraphics[width=0.45\textwidth]{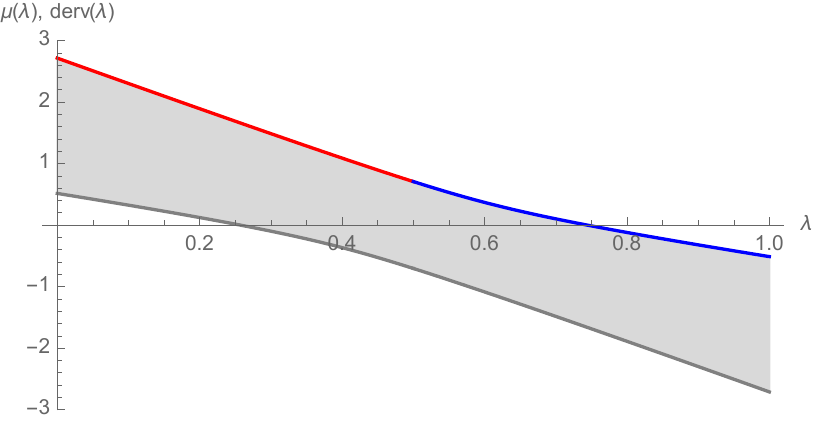}
\caption{Left: Space of generators $\mf Q$ parametrized as in Proposition~\ref{prop:general-parametrization}.
The boundary, as determined in Corollary~\ref{coro:boundary-gens}, is highlighted in red (upper part) and blue (lower part), with the relevant part solid and the rest dotted.
Right: Differential inclusion $\lambda\mapsto\derv(\lambda)$ with the optimal derivatives $\mu$ highlighted in red and blue. They are determined using Lemma~\ref{lemma:analytic-params}.
The Lindblad terms used are the same as in Figure~\ref{fig:example-space-of-lines}.
}
\label{fig:general-lines-points}
\end{figure}

\subsection{Optimal solution}

In order to find the optimal derivative function $\mu$ we first find the boundary of $\mf Q$ by determining the maximal and minimal values of $z\mapsto F(x,z)$ for all $x\in[0,1/2]$, cf.\ Proposition~\ref{prop:polarity}. 
Consider the simplified function
\begin{align} \label{eq:G-phi}
G_\phi(\xi,\zeta)=(1-\xi)\sin(4\pi \zeta+\phi-\pi/2)-\xi\sin(2\pi\zeta), \quad \xi,\zeta\in[0,1].
\end{align}
See Appendix~\ref{app:g-phi} for the relevant properties of this function. 
In particular see Figure~\ref{fig:contour-g} for some plots for different $\phi$. 

\begin{lemma} \label{lemma:to-G-phi}
Assume that $r_1,r_2\neq0$ and let $x\in(0,1/4)$ be given.\footnote{By symmetry it suffices to consider the right side of $\mf Q$.} Let $\xi = \frac1{1+\frac{r_1}{r_2}\tan(2\pi x)}\in(0,1)$. Then it holds for all $z\in\R$ that
\begin{align*}
z\in\underset{\tilde z}{\argmax}\ F(x,\tilde z)
\iff 
\zeta \in \underset{\tilde \zeta}{\argmax}\ G_\phi(\xi,\tilde \zeta)
\end{align*}
where $\zeta=z+\phi_2/(2\pi)$. The same statement holds after replacing $\argmax$ by $\argmin$.
\end{lemma}

\begin{proof}
We compute
\begin{align*}
\underset{\tilde z\in[0,1]}{\argmax}\ F(x,\tilde z)
&=
\underset{\tilde z\in[0,1]}{\argmax} \
r_1\sin(2\pi x)\sin(4\pi \tilde z+\phi_1)
-r_2\cos(2\pi x)\sin(2\pi \tilde z+\phi_2)
\\&=
\underset{\tilde z\in[0,1]}{\argmax}\
\sin(4\pi \tilde z+\phi_1)
-\frac{r_2}{r_1}\cot(2\pi x) \sin(2\pi \tilde z+\phi_2)
\\&=
\underset{\tilde z\in[0,1]}{\argmax}\
\sin(4\pi \tilde z+\phi_1)
-\frac{\xi}{1-\xi} \sin(2\pi\tilde z+\phi_2)
\\&=
\underset{\tilde z\in[0,1]}{\argmax}\
(1-\xi)\sin(4\pi \tilde z+\phi_1)-\xi\sin(2\pi \tilde z+\phi_2)
\\&=
\underset{\tilde z\in[0,1]}{\argmax}\
G_\phi(\xi, \tilde z + \phi_2/(2\pi)),
\end{align*}
and the computation remains valid after replacing $\argmax$ by $\argmin$.
\end{proof}

\begin{remark} \label{rmk:degen}
If $r_1=0$ or $r_2=0$ or both, then the values of $z$ which maximize or minimize $F(x,z)$ can be chosen independently of $x$. We will call such systems \emph{degenerate}. 
Note that for the  Bloch equations (Appendix~\ref{subsec:bloch-eq}) it holds that $r_1=r_2=0$ and for rank one systems (cf.\ \cite[Sec.~V]{OptimalCooling24}) it holds that $r_2=0$. 
For the sake of brevity we will not treat degenerate systems in detail, and leave this as an exercise to the reader.
\end{remark}

\begin{corollary}
Assume that $\phi\in(-\pi,0)\cup(0,\pi)$ and define the function
\begin{align*}
x^\star(z) = \frac{1}{2\pi}\operatorname{arccot}\left({2\,\frac{r_1}{r_2}\,\frac{\cos(4\pi z+\phi_1)}{\cos(2\pi z+\phi_2)}}\right),
\end{align*}
and the intervals
\begin{align*}
I^+=\begin{cases}
[\tfrac34-\frac{\phi}{4\pi}-\tfrac{\phi_2}{2\pi},\, \tfrac34-\tfrac{\phi_2}{2\pi}] & \text{ if } \phi>0\\
[\tfrac34-\frac{\phi_2}{2\pi},\, \tfrac34-\tfrac{\phi}{4\pi}-\tfrac{\phi_2}{2\pi}] &\text{ if } \phi<0,
\end{cases} \quad
I^-=\begin{cases}
[\tfrac14-\tfrac{\phi_2}{2\pi},\tfrac14+\tfrac{\pi-\phi}{4\pi}-\tfrac{\phi_2}{2\pi}] &\text{ if } \phi>0\\
[-\tfrac{\phi}{4\pi}-\tfrac{\phi_2}{2\pi},\tfrac14-\tfrac{\phi_2}{2\pi}] &\text{ if } \phi<0\,.
\end{cases}
\end{align*}
Let $x^+$ and $x^-$ denote the restrictions of $x^\star$ to $I^+$ and $I^-$ respectively. These functions are bijective onto $[0,1/4]$, and it holds that
\begin{align*}
(x^+)^{-1}(x) = \underset{\tilde z}{\argmax}\ F(x,\tilde z), \quad
(x^-)^{-1}(x) = \underset{\tilde z}{\argmin}\ F(x,\tilde z)
\end{align*}
for all $x\in(0,1/4)$. 
Here the $\argmax$ and $\argmin$ are unique for all $x\in(0,1/4)$. 
For $x=1/4$ there is a spurious second solution  and for $x=0$ we get $F(0,z)=(\Delta,\Sigma)$. 
\end{corollary}

\begin{proof}
This follows directly from Lemma~\ref{lemma:to-G-phi} and Lemma~\ref{lemma:G-phi-opt}.
\end{proof}

Unfortunately it seems that the inverses of the functions $x^+$ and $x^-$ defined in the previous lemma cannot be computed analytically. However we can obtain analytical expressions parametrized by $z$.

\begin{corollary} \label{coro:boundary-gens}
Assume that $\phi\in(-\pi,0)\cup(0,\pi)$ and define the path
\begin{align*}
\gamma(z)=(\Delta\cos(2\pi x^\star(z)), F(x^\star(z),z)),
\end{align*}
and let $\gamma^+$ and $\gamma^-$ be the restrictions to $I^+$ and $I^-$ respectively. Then $\gamma^+$ is a parametrization of the upper boundary of the right half of the space of generators $\mf Q$, and analogously $\gamma^-$ parametrizes the lower boundary of the right half.
It follows that the boundary point $\gamma(z)$ for $z\in I^+\cup I^-$ can be obtained using the unitary $U_{x^\star(z),z}$.
\end{corollary}

The intervals $I^+$ and $I^-$ are still too large, since on these intervals $\gamma$ parametrizes part of the boundary of the space of generators which are not relevant for optimal control. 
Indeed, since $\lambda\in[0,1]$, we are only interested in the values of $z$ where $\gamma_1'(z)/\gamma_2'(z)\in[-1,1]$. 
This parameter region can be computed numerically and we will denote the corresponding closed parameter intervals $J^+\subseteq I^+$ and $J^-\subseteq I^-$.

With this we can apply the results of Section~\ref{sec:optimal-controls} to numerically determine the optimal path through the Bloch ball and the corresponding control functions of the full control system~\eqref{eq:bilin}, cf.~Figure~\ref{fig:general-opt-path}.

\begin{figure}[thb]
\centering
\includegraphics[align=c,width=0.35\textwidth, trim=1cm 0.5cm 0.7cm 0cm,clip]{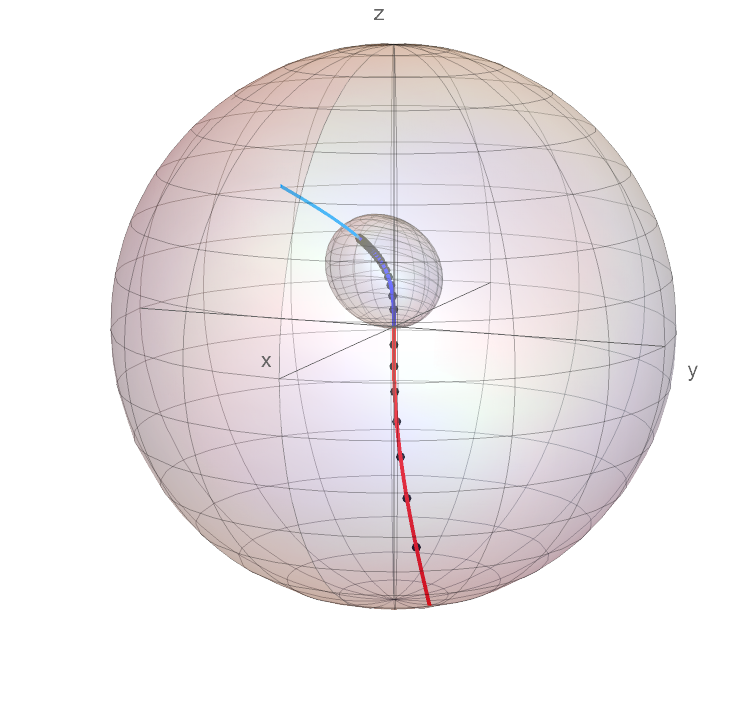}
\includegraphics[align=c,width=0.55\textwidth]{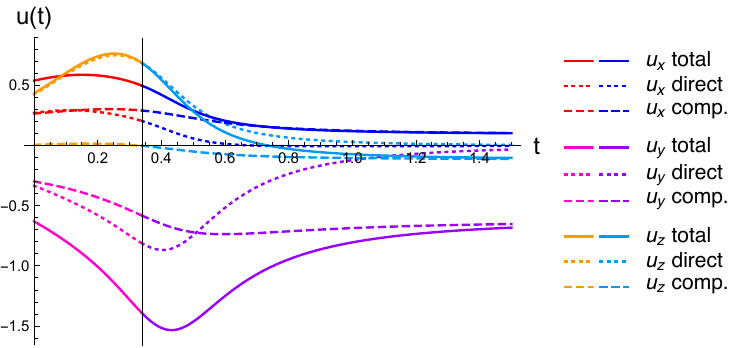}
\caption{
Left: Optimal path for heating (red) and cooling (blue) in the Bloch ball together with the set of stabilizable states. The light blue part indicates the part of the optimal path which is not reachable from the interior as it lies beyond the purest stabilizable state. The black dots on the path are equally spaced in time.
Right: Optimal control functions $u_x,u_y$ and $u_z$ for heating (left part) and cooling (right part) and their contributing direct and compensating terms.
The Lindblad terms used are the same as in Figure~\ref{fig:example-space-of-lines}.
}
\label{fig:general-opt-path}
\end{figure}

\section{Integral systems} \label{sec:integral}

In the previous section we had to exclude the cases where $\phi$, as defined in Section~\ref{sec:gen-param}, is an integer multiple of $\pi$.
It turns out that in these cases the general parametrization of $\mf Q$ obtained in the previous section simplifies considerably.
We call such systems integral and we explore their properties in this section.

\begin{definition}
A system is \emph{integral} if it is non-degenerate (cf.~Remark~\ref{rmk:degen}) and $\phi=k\pi$ for some $k\in\Z$. 
We call $p=(-1)^k$ the \emph{parity} of the system, and we say that the system is \emph{even} or \emph{odd} if $k$ is even or odd respectively. 
\end{definition}
\noindent Note that due to Lemma~\ref{lemma:well-def} integral systems and their parity are well-defined.

\begin{example}
We say that a Lindblad generator $-L\in\wkl$ is \emph{real} if there exists a choice of real Lindblad terms $V_k$.
If all Lindblad terms are real, and the system is non-degenerate, then the system is integral.
Indeed, since all $V_k$ are real, the sum $\sum_{k=1}^r[V_k,V_k^*]$ is real and symmetric and hence can be orthogonally diagonalized. 
Hence $r_1e^{\iu\phi_1}$ is imaginary and $r_2e^{\iu\phi_2}$ is real, so that $\phi_1+\pi/2$ and $2\phi_2$ are integer multiples of $\pi$.
Note that both even and odd systems can be obtained in this fashion, and they are separated in $\wkl$ by degenerate systems.
\end{example}
\noindent Since they are easy to generate, we will use real systems for the plots shown in this section.

\begin{figure}[thb]
\centering
\includegraphics[width=0.45\textwidth]{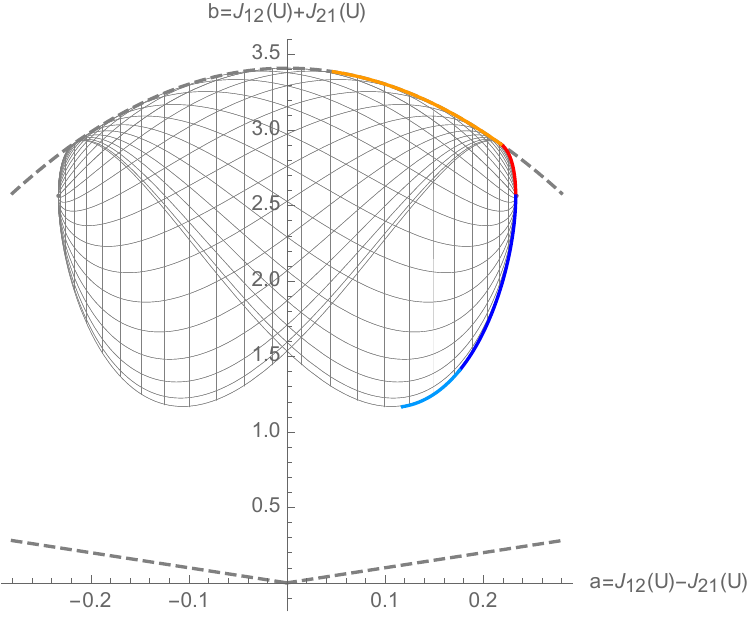}
\includegraphics[width=0.45\textwidth]{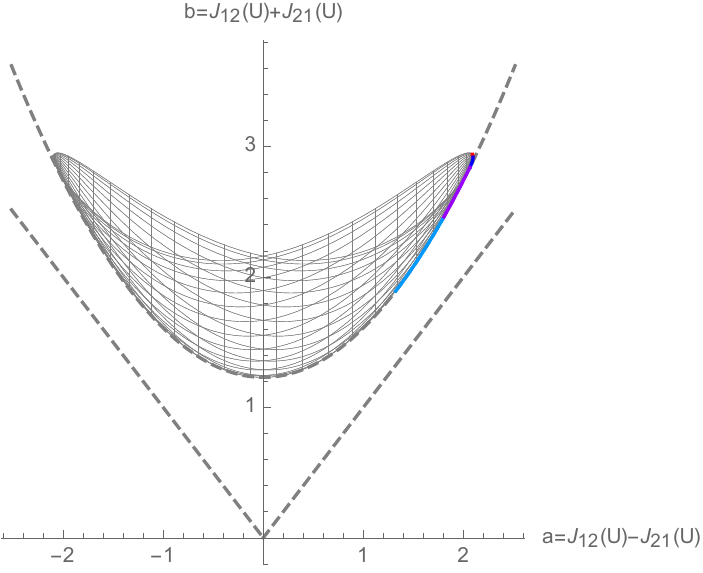}
\caption{
Space of generators $\mf Q$ for two randomly chosen real systems. 
The system on the left is odd, and the system on the right is even.
The boundary, computed using Lemma~\ref{lemma:integral-boundary}, features a parabolic part, located at the top in the odd case and at the bottom in the even case.
The Lindblad terms for the odd case are
$\big(\begin{smallmatrix} 0.3 & -1.0 \\ -0.6 & 0.9\end{smallmatrix}\big),
\big(\begin{smallmatrix} 0.8 & 0.2 \\ 0.8 & 0.8\end{smallmatrix}\big),
\big(\begin{smallmatrix} 0.1 & -0.4 \\ 0.0 & -0.8\end{smallmatrix}\big)$, and
$\big(\begin{smallmatrix} -0.3 & 0.6 \\ 0.6 & 0.4\end{smallmatrix}\big)$.
For the even case they are
$\big(\begin{smallmatrix} 0.5 & 0.1 \\ -0.4 & -0.1 \end{smallmatrix}\big),
\big(\begin{smallmatrix} 0.0 & 0.5 \\ -0.9 & 0.0 \end{smallmatrix}\big),
\big(\begin{smallmatrix} 0.9 & -0.3 \\ 0.7 & 0.6 \end{smallmatrix}\big)$, and
$\big(\begin{smallmatrix} 0.4 & -0.2 \\ -1.0 & 0.0 \end{smallmatrix}\big)$.}
\label{fig:real-systems}
\end{figure}

\subsection{Optimal solution}

Using Lemma~\ref{lemma:G-phi-opt-integral} we can give analytic expressions for the boundary of the space of generators $\mf Q$ for integral systems.

\begin{lemma} \label{lemma:integral-boundary}
Set $\tilde x = \frac{1}{2\pi}\arctan(\frac{r_2}{4r_1})$. For $\phi=0$, we obtain for every $x\in[0,1/4]$ that
\begin{align*}
\max_z F(x,z) &= \Sigma+\delta\sin(2\pi x)^2 + \sin(2\pi x)(r_1\sin(2\pi x) + r_2\cos(2\pi x)) \\
\min_z F(x,z) &= \begin{cases}
\Sigma+\delta\sin(2\pi x)^2 + \sin(2\pi x)(r_1\sin(2\pi x) - r_2\cos(2\pi x))
\quad &\text{if }x\in[0,\tilde x] \\
\Sigma+\sin(2\pi x)^2(\delta-r_1)-\cos(2\pi x)^2\frac{r_2^2}{8r_1}
\quad &\text{if }x\in[\tilde x,1/4],
\end{cases}
\end{align*}
and analogously for $\phi=\pi$ we obtain that
\begin{align*}
\max_z F(x,z) &= \begin{cases}
\Sigma+\delta\sin(2\pi x)^2 + \sin(2\pi x)(-r_1\sin(2\pi x) + r_2\cos(2\pi x))
\quad &\text{if }x\in[0,\tilde x] \\
\Sigma+\sin(2\pi x)^2(\delta+r_1)+\cos(2\pi x)^2\frac{r_2^2}{8r_1}
\quad &\text{if }x\in[\tilde x,1/4]
\end{cases} \\
\min_z F(x,z) &= \Sigma+\delta\sin(2\pi x)^2 + \sin(2\pi x)(-r_1\sin(2\pi x) - r_2\cos(2\pi x)).
\end{align*}
Furthermore these values can be obtained using the unitary $U_{x,z}=e^{\iu\pi z^\pm(x)\sigma_z} e^{\iu\pi x\sigma_x}$ where
\begin{align*}
z^+(x) = \tfrac34-\tfrac{\phi_2}{2\pi}, \quad
z^-(x) = \begin{cases}
\tfrac14-\tfrac{\phi_2}{2\pi} &\text{if }x\in[0,\tilde x] \\
\tfrac1{2\pi}\arcsin(\tfrac{r_2}{4r_1\tan(2\pi x)})-\tfrac{\phi_2}{2\pi} &\text{if }x\in[\tilde x,1/4],
\end{cases}
\end{align*}
in the even case and
\begin{align*}
z^-(x) = \tfrac14-\tfrac{\phi_2}{2\pi}, \quad
z^+(x) = \begin{cases}
\tfrac34-\tfrac{\phi_2}{2\pi} &\text{if }x\in[0,\tilde x] \\
\tfrac12+\tfrac1{2\pi}\arcsin(\tfrac{r_2}{4r_1\tan(2\pi x)})-\tfrac{\phi_2}{2\pi} &\text{if }x\in[\tilde x,1/4],
\end{cases}
\end{align*}
in the odd case.
\end{lemma}

\begin{proof}
Using Lemma~\ref{lemma:to-G-phi} and Lemma~\ref{lemma:G-phi-opt-integral} we immediately find the optimal value of $z$ as a function of $x$. Then we find the switching point
$\xi = 4/5 \iff x = \tfrac{1}{2\pi}\arctan(\tfrac{r_2}{4r_1})$
and we denote this value by $\tilde x$.
For the even case $\phi=0$, or equivalently $\phi_1-2\phi_2=-\pi/2$, we compute
\begin{align*}
\sin(4\pi z^+(x)+\phi_1)=\sin(3\pi-2\phi_2+\phi_1)=1 \quad &\text{ for } x\in[0,1/4]\\
\sin(2\pi z^+(x)+\phi_2)=\sin(3\pi/2-\phi_2+\phi_2)=-1 \quad &\text{ for } x\in[0,1/4]\\
\sin(4\pi z^-(x)+\phi_1)=1
\quad &\text{ for }x\in[0,\tilde x]\\
\sin(2\pi z^-(x)+\phi_2)=1
\quad &\text{ for }x\in[0,\tilde x],
\end{align*}
and using that $1/(4(1/\xi-1))=r_2/(4r_1)\cot(2\pi x)$ we get
\begin{align*}
\sin(4\pi z^-(x)+\phi_1)
&=-\cos(2\arcsin(r_2/(4r_1)\cot(2\pi x)))\\
&=2(r_2/(4r_1)\cot(2\pi x))^2-1
\quad &\text{ for }x\in[\tilde x,1/4]\\
\sin(2\pi z^-(x)+\phi_2)
&= r_2/(4r_1)\cot(2\pi x)
\quad &\text{ for }x\in[\tilde x,1/4],
\end{align*}
and hence we find
\begin{align*}
F(x,z^-(x))=\Sigma+\sin(2\pi x)^2(\delta-r_1)-\cos(2\pi x)^2\frac{r_2^2}{8r_1}
\quad &\text{ for }x\in[\tilde x,1/4].
\end{align*}
The computations for the odd case $\phi=\pi$ are analogous and this yields the result.
\end{proof}
\noindent The boundary parametrization is shown in Figure~\ref{fig:real-systems}, and the resulting optimal derivative function $\mu$ as well as the optimal path and the set of stabilizable states in the Bloch sphere are presented in Figure~\ref{fig:real-derv} and Figure~\ref{fig:real-bloch} respectively. 

\begin{figure}[thb]
\centering
\includegraphics[width=0.45\textwidth]{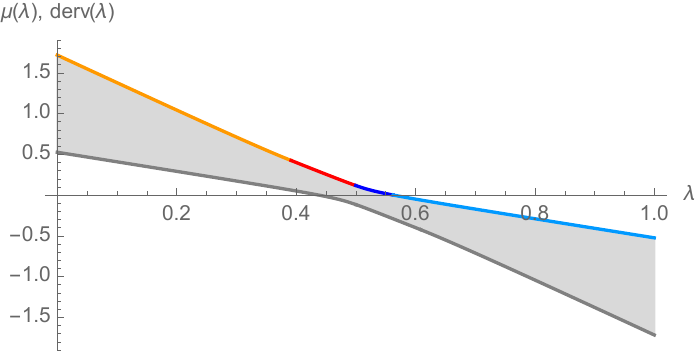}
\includegraphics[width=0.45\textwidth]{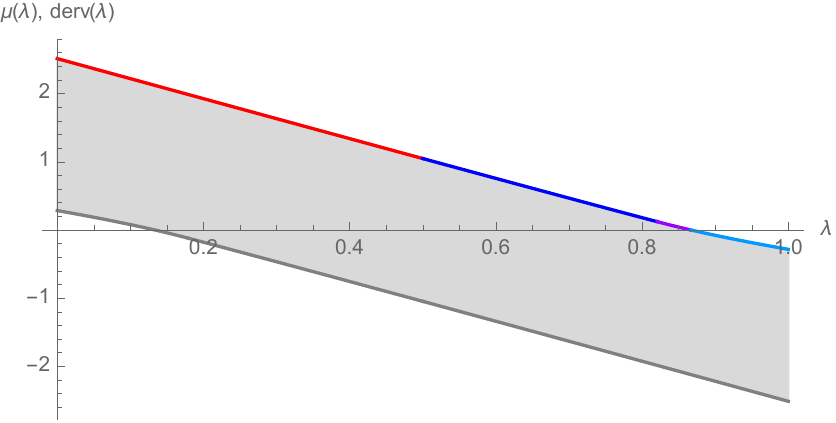}
\caption{For the same two real systems we plot the set-valued function $\derv$ of achievable derivatives and the optimal derivatives $\mu$, which can be computed using Lemma~\ref{lemma:analytic-params}. The colors used correspond to those of Figure~\ref{fig:real-systems}.}
\label{fig:real-derv}
\end{figure}

\begin{figure}[thb]
\centering
\includegraphics[width=0.45\textwidth]{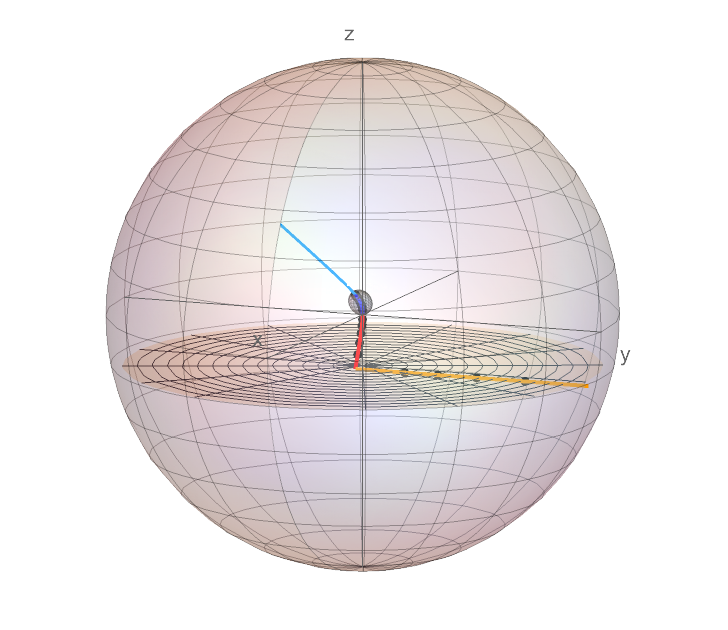}
\includegraphics[width=0.45\textwidth]{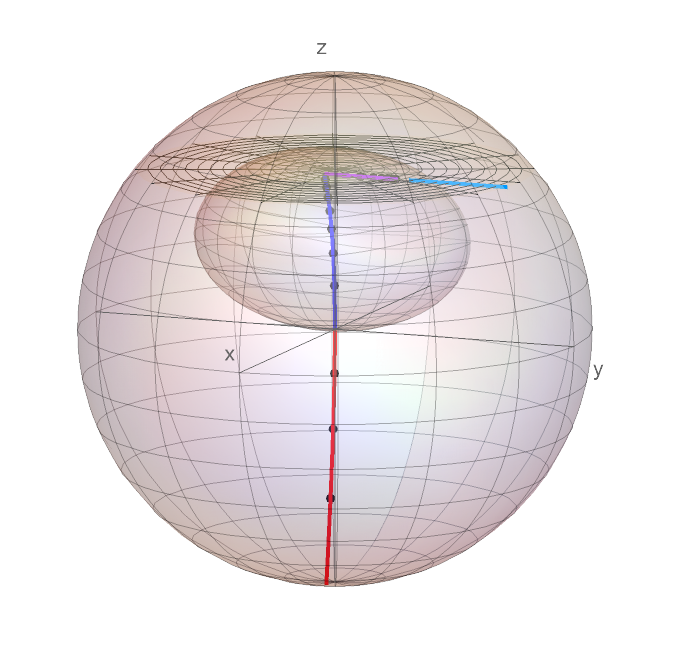}
\caption{For the same two real systems we plot the optimal path through the Bloch ball and the set of stabilizable states. Note that the parabolic parts of the boundary of $\mf Q$ correspond to horizontal parts in the Bloch ball, reminiscent of the magic plane in the Bloch equations case, cf.\ Section~\ref{subsec:bloch-eq}. The rest of the path lies in the $xz$-plane. The colors used correspond to those of Figure~\ref{fig:real-systems}.}
\label{fig:real-bloch}
\end{figure}

\subsection{Coolable systems} \label{sec:coolable}

In~\cite[Thm.~4.7]{LindbladReduced23} coolability of Markovian quantum systems with fast unitary control was characterized.
As a consequence, in the qubit case the system is asymptotically coolable if and only if the Lindblad terms $V_k$ can be simultaneously unitarily triangularized without being simultaneously diagonal, see also~\cite[Thm.~5.2]{Rooney12}.
For non-unital rank one systems this is always satisfied, and such systems were studied in detail in~\cite[Sec.~V]{OptimalCooling24}.

\begin{lemma}
Coolable systems are either degenerate or odd.
\end{lemma}

\begin{proof}
As usual we choose the Lindblad terms $V_k$ traceless and the basis such that $\sum_{k=1}^r[V_k,V_k^*]$ is diagonal.
Then the $V_k$ have the form
$V_k =
\big(\begin{smallmatrix}
u_k&v_k\\w_k&-u_k
\end{smallmatrix}\big)$,
where $u_k,v_k,w_k\in\C$ for $k=1,\ldots,r$.
Let $u,v,w\in\C^r$ be the vectors with coefficients $u_k$, $v_k$ and $w_k$ respectively.
We will use the shorthand notation $(u,v)=\sum_{k=1}^r\overline{u}_kv_k$ and $|v|=\sqrt{(v,v)}$.
Diagonality of $\sum_{k=1}^r[V_k,V_k^*]$ is equivalent to $(v,u)=(u,w)$.
It is easy to show that one can always modify the basis (while keeping $\sum_{k=1}^r[V_k,V_k^*]$ diagonal) such that $(v,w)\geq0$ and we will make this assumption.
If $(v,w)=0$ then $r_1=0$ and the system is degenerate, hence assume that $(v,w)>0$ and thus $\phi_1=\tfrac\pi2$.
Now consider a common eigenvector of all $V_k$. If it is proportional to $(1,0)^\top$ or $(0,1)^\top$ it is easy to see that $r_2=0$.
Hence assume that the eigenvector is of the form $(1,\beta)$ with $\beta\neq0$.
We will show that indeed $\beta\in\R$.
From the eigenvalue equation we obtain that $2u_k=\tfrac{w_k}\beta-\beta v_k$ for all $k=1,\ldots,r$.
By taking the inner product with $v$ and $w$ and using $(v,u)=(u,w)$ we find that 
$\beta|v|^2+\tfrac1{\overline\beta}|w|^2 = (\tfrac1\beta+\overline\beta)(v,w)$.
By considering the complex argument of each side it is clear that $\beta$ must be real.
But then it follows from the above that $2(v,u)=\tfrac1\beta(v,w)-\beta|v|^2$ and so $(v,u)$ is real as well.
Thus $\phi_2=n\pi$ is an integer multiple of $\pi$.
Thus $\phi=\phi_1-2\phi_2+\tfrac\pi2=(1-2n)\pi$ and the system is odd as desired.
\end{proof}

Since non-degenerate coolable systems are integral, the results of the previous section still apply.
However, using the simultaneous triangular form of the Lindblad terms another parametrization can be obtained.
Choosing all $V_k$ traceless and an appropriate basis we may assume that they are of the form
$V_k =
\big(\begin{smallmatrix}
u_k&v_k\\0&-u_k
\end{smallmatrix}\big)$,
where $u_k, v_k\in\C$ are arbitrary for $k=1,\ldots,r$. Let $u,v\in\C^r$ be the vectors with coefficients $u_k$ and $v_k$, then we define
\begin{alignat*}{3}
c_1(x) &= |v|^2\cos(2\pi x),\quad     & c_3(x) &= 2|u|^2\sin^2(2\pi x) + |v|^2(\cos^2(2\pi x)+1)/2,\\
c_2(x) &= -2|(u,v)|\sin(2\pi x),\quad & c_4(x) &= -|(u,v)|\sin(4\pi x).
\end{alignat*}
where $(u,v)=\sum_{k=1}^r\overline{u}_kv_k$ and $|v|=\sqrt{(v,v)}$. 

\begin{lemma} \label{lemma:cool-param}
For a coolable system in the form described above, the space of generators can be parametrized as
$\{(c_1(x)+c_2(x)\sin(2\pi z),c_3(x)+c_4(x)\sin(2\pi z)):\,x\in[0,1/2],\, z\in[0,1)\}$,
where each point is obtained using the unitary $U_{x,z}$.
The non-parabolic part of the boundary is achieved by $z=\tfrac14$.
Moreover, the parabolic segment lies on a parabola which is tangent to the bisectors and has the form
$a\mapsto\Sigma+\delta+r_2/2 + \frac{1}{4(\Sigma+\delta+r_2/2)}\ a^2$.
\end{lemma}

\begin{proof}
This follows directly from the computations in Appendix~\ref{app:parametrization}.
\end{proof}

\begin{figure}[thb]
\centering
\includegraphics[width=0.45\textwidth]{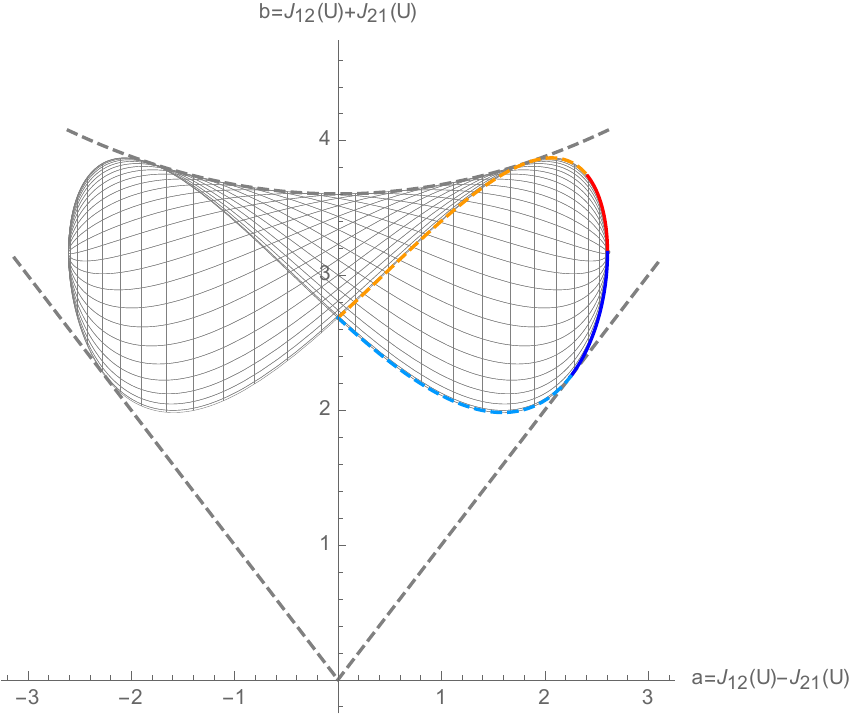}
\includegraphics[width=0.45\textwidth]{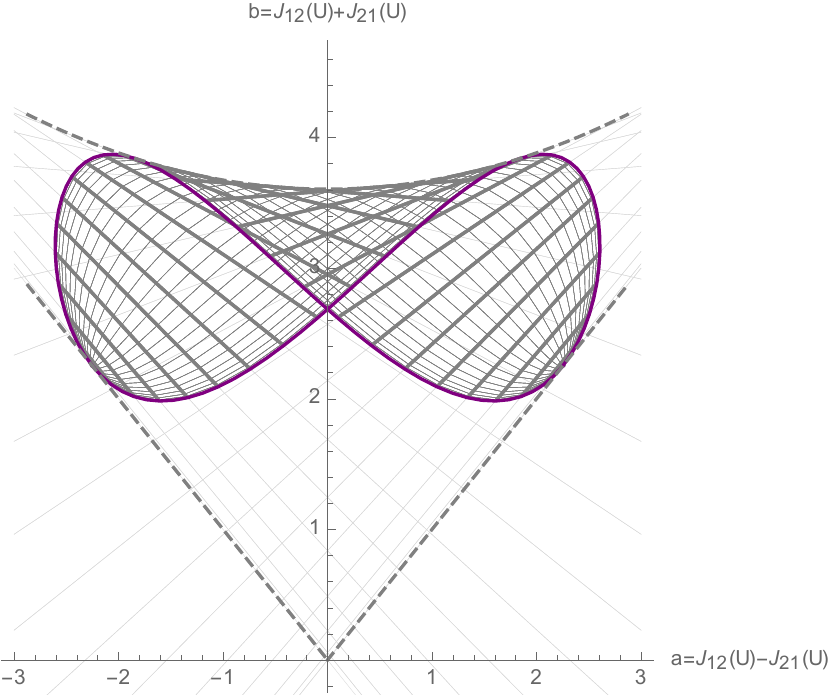}
\caption{Two different parametrizations of the space of generators for a coolable system. 
The usual parametrization is on the left, and the parametrization of Lemma~\ref{lemma:cool-param} is on the right.
The Lindblad terms are given by
$\big(\begin{smallmatrix} 0.7-0.5\iu & 0.6-0.6\iu \\ 0.0 & -0.8+0.9\iu \end{smallmatrix}\big),
\big(\begin{smallmatrix} 0.8-0.3\iu & -1.0 \\ 0.0 & 0.1 \end{smallmatrix}\big)$, and
$\big(\begin{smallmatrix} 0.3-0.2\iu & -0.2 + 0.7\iu \\ 0.0 & 0.2-0.6\iu \end{smallmatrix}\big)$.
}
\label{fig:real-systems2}
\end{figure}

\section{Acknowledgments}

I would like to thank Frederik vom Ende for his valuable feedback.
The project was funded i.a.\ by the Excellence Network of Bavaria under ExQM, by {\it Munich Quantum Valley} of the Bavarian State Government with funds from Hightech Agenda {\it Bayern Plus}.

\appendix

\section{Special systems} \label{sec:special-systems}

In this section we study some highly structured systems, namely unital systems and the Bloch equations.
The case of rank one systems, i.e.\ those defined by a single Lindblad term, has been treated in detail by the author in~\cite[Sec.~V]{OptimalCooling24}.

\subsection{Unital systems} \label{app:unital}

A simple but broad class of examples is given by \emph{unital} systems, which are defined by $-L(\mathds1)=\sum_{k=1}^r [V_k,V_k^*]=0$. 
This condition is independent of the Hamiltonian part of $-L$, and hence also of the control Hamiltonians. 
Unital systems in arbitrary (finite) dimension have been addressed in~\cite[Sec.~5]{LindbladReduced23}.
The qubit case presented here allows for even stronger results.
Unital channels were also studied in~\cite[Sec.~IV,V]{Mukherjee13}.

\begin{lemma} \label{lemma:unital}
The following are equivalent:
\begin{enumerate*}[(i)]
\item \label{it:unital} $-L$ is unital,
\item \label{it:unital-u} the optimal derivative function satisfies $\mu(1/2)=0$,
\item \label{it:unital-sol} the space of generators satisfies $\mf Q\subset\{(0,y):y\geq0\}$.
\end{enumerate*}
\end{lemma}

\begin{proof}
By Lemma~\ref{lemma:mu-basic}, $\mu(1/2)=0$ if and only if $\sum_{k=1}^r [V_k,V_k^*]=0$. 
This shows the equivalence of~\ref{it:unital} and~\ref{it:unital-u}. 
Now assume~\ref{it:unital-u}, then every line in $\mf Q$ passes through $(1/2,0)$ due to the central symmetry.
This is equivalent to $J_{12}(U)=J_{21}(U)$ for all $U\in\SU(2)$ and hence to~\ref{it:unital-sol}.
\end{proof}
\noindent Condition~\ref{it:unital-u} implies that the graph of $\derv$ is a cone with origin $(1/2,0)$, as illustrated in Figure~\ref{fig:unital}.

\begin{figure}[thb]
\centering
\includegraphics[width=0.45\textwidth]{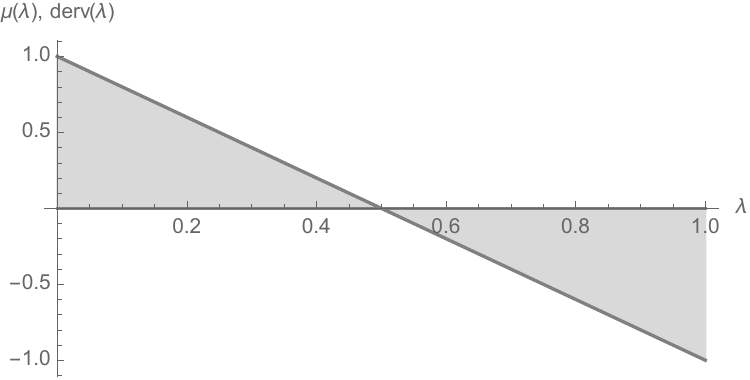}
\includegraphics[width=0.45\textwidth]{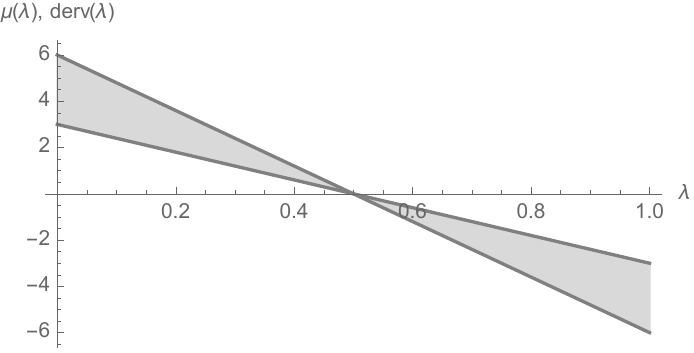}
\caption{The graph of $\derv$ for two unital systems. As shown in Lemma~\ref{lemma:unital}, all lines pass through $(1/2,0)$, and this is indeed characteristic of unital systems. The systems used are defined by $\sqrt{\gamma_k}\sigma_k$ for $k\in{x,y,z}$. Since the Pauli matrices are normal, such systems are always unital. 
Left: The system is defined by ${\gamma_x}=1$ and ${\gamma_y}={\gamma_z}=0$. So it is of rank one and by Lemma~\ref{lemma:unital-stab} it is unital stabilizable. By Lemma~\ref{lemma:unital-stab-u0} it holds that $\derv(0)=[0,1]$. 
Right: The system is defined by ${\gamma_x}=4$, ${\gamma_y}=2$, and ${\gamma_z}=1$. Hence by Proposition~\ref{proposition:unital-systems} it holds that $\derv(0)=[3,6]$.}
\label{fig:unital}
\end{figure}
\noindent It is clear that for unital systems the stabilizable region is either $\{1/2\}$ or $[0,1]$.
Moreover $\derv$, $\mu$, and $\mf Q$ are completely described by the minimal and maximal values of $\derv(0)$.

\begin{lemma} \label{lemma:unital-stab}
Let $-L$ be an arbitrary Lindblad generator and let $V_k$ be a corresponding family of Lindblad terms. Then the following are equivalent:
\begin{enumerate*}[(i)]
\item \label{it:unit-stab-normal} all $V_k$ are normal and commute with each other,
\item \label{it:unit-stab-origin} the space of generators satisfies $(0,0)\in\mf Q$,
\item \label{it:unit-stab-regions} the system is unital and the stabilizable region is all of $[0,1]$.
\end{enumerate*}
We call such systems \emph{unital stabilizable}.
\end{lemma}

\begin{proof}
Since~\ref{it:unit-stab-normal} is equivalent to the existence of a unitary $U\in\SU(2)$ simultaneously diagonalizing all $V_k$, it is also equivalent to~\ref{it:unit-stab-origin}. 
Now assume the conditions above. 
This implies that all of $[0,1]$ is stabilizable and that the system is unital.
This shows~\ref{it:unit-stab-regions}. 
Finally assume~\ref{it:unit-stab-regions}. The system is unital and so $[0,1]$ is stabilizable only if $(0,0)\in\mf Q$. 
This concludes the proof.
\end{proof}
\noindent Clearly for a unital stabilizable system, $\mu$ is completely characterized by its value at $0$. The following result yields a simple way to determine this value.

\begin{lemma} \label{lemma:unital-stab-u0}
Consider a unital stabilizable system. 
Then $\mu(0)=\frac14\sum_{k=1}^r|\lambda_1(V_k)-\lambda_2(V_k)|^2$ where the $\lambda_i(V_k)$ denote the eigenvalues of $V_k$.
\end{lemma}

\begin{proof}
By Lemma~\ref{lemma:unital-stab} all Lindblad terms $V_k$ are simultaneously unitarily diagonalizable. 
Using a unitary change of basis, we may assume that they are indeed diagonal. 
Moreover we may assume that they are traceless. 
Hence all Lindblad terms are multiples of $\sigma_z$.
Applying a unitary reshuffling to the $V_k$ one may assume that all but one $V_k$ are zero.
Performing this calculation yields the desired result.
\end{proof}

Now we will completely describe $\mu$ and $\mf Q$ for arbitrary unital systems. 
For this let $C\in\C^{3,3}$ denote the Kossakowski matrix (cf.~\cite[p.~121]{BreuPetr02}) of the system with respect to the Pauli basis $\{\sigma_x,\sigma_y,\sigma_z\}$.

\begin{lemma}
The system is unital if and only if $C$ is real.
\end{lemma}

\begin{proof}
We will give a sketch of the proof.
Using the general ``non-diagonal'' form of the Lindblad equation (cf.~\cite[p.~121]{BreuPetr02}) with respect to the Pauli basis one finds by a simple computation
that the generator is unital, i.e.\ $-L(\id)=0$, if and only if the Kossakowski matrix $C$ is symmetric, or equivalently, real. 
\end{proof}

\begin{proposition} \label{proposition:unital-systems}
Consider a unital system.
Let $\gamma_1\geq\gamma_2\geq\gamma_3$ denote the eigenvalues of $C$.
Then $\derv(0)=[\gamma_2+\gamma_3,\gamma_1+\gamma_2]$.
\end{proposition}

\begin{proof}
Applying a unitary basis transformation to the qubit system changes the Kossakowski matrix via a corresponding orthogonal transformation.
Since $C$ is real, it can be orthogonally diagonalized and hence, without loss of generality, we can assume that the Lindblad terms are $\sqrt{\gamma_1}\sigma_x$, $\sqrt{\gamma_2}\sigma_y$, and $\sqrt{\gamma_3}\sigma_z$.
The result then follows easily from the computation in Appendix~\ref{app:parametrization}.
\end{proof}

\subsection{Bloch equations} \label{subsec:bloch-eq}

Another family of simple Lindblad generators for the qubit is given by those which have a rotation symmetry about some axis, which reduces the dimension of the problem.
Without loss of generality we assume that the symmetry is about the $z$-axis. 
Such generators correspond to the well-known Bloch equations. 
It turns out that such systems can be solved analytically, and we will do so in detail in this section. 
In Lemma~\ref{lemma:magic-plane} we recover a known result from~\cite{Lapert10} about the so-called magic plane for optimal heating and in Lemma~\ref{lemma:bloch-stab} we recover the steady state ellipsoid from~\cite{Lapert13}.
The optimal controls can also be determined, as was done in~\cite[Sec.~6]{Reduced23}.
A special case of the Bloch equations was also considered in detail in~\cite[Sec.~III]{Mukherjee13}.

The Bloch equations are equivalent to the Lindblad generator defined by the Lindblad terms $\sqrt{\gamma_+}\sigma_+$, $\sqrt{\gamma_-}\sigma_-$ and $\sqrt{\gamma_z}\sigma_z$ where $\sigma_\pm=(\sigma_x\pm\iu\sigma_y)/2$ and $\gamma_+,\gamma_-,\gamma_z\geq0$. The case considered in~\cite[Sec.~III]{Mukherjee13} corresponds to the case $\gamma_z=0$.
For convenience we introduce the following parameters:
\begin{align*}
\Delta = |\gamma_+-\gamma_-|, \quad
\Sigma = \gamma_++\gamma_-, \quad
\delta = 2\gamma_z - \Sigma/2\,.
\end{align*}
Indeed, the values of $\Delta$, $\Sigma$, and $\delta$ defined here correspond to those defined in the general case in Section~\ref{sec:general-systems} with $r_1=r_2=0$.
By a change of basis we may and will assume without loss of generality that $\gamma_+\geq\gamma_-$.

\begin{remark} 
The Bloch equations for a single qubit are often written in the following form, see for instance~\cite[Sec.~5.5]{AM11},
\begin{align*}
\dot{\bf M}=\gamma {\bf M}\times {\bf B}-R({\bf M}-{\bf M}_\beta),\quad R=\diag(T_2^{-1},T_2^{-1},T_1^{-1})\,,
\end{align*}
where ${\bf M}$ is the spin magnetization of the system, $\gamma$ is the gyromagnetic ratio, ${\bf B}=(0,0,B)$ is the magnetic field, and $R$ is the relaxation matrix with $T_1$ the longitudinal and $T_2$ the transversal relaxation time. Furthermore ${\bf M}_\beta=(0,0,M_\beta)$ is the steady state magnetization and satisfies $2M_\beta/(\gamma\hbar)\in[-1,1]$. Then we have the relations
\begin{gather*}
T_1^{-1} = \gamma_+ + \gamma_- = \Sigma,\quad
T_2^{-1} = (\gamma_+ + \gamma_- + 4\gamma_z)/2
= \Sigma + \delta \\
2M_\beta/(\gamma\hbar) = (\gamma_+-\gamma_-)/(\gamma_++\gamma_-) = \Delta/\Sigma\,,
\end{gather*}
Note that the famous relation $2T_1\geq T_2$ is equivalent to the non-negativity of the relaxation rates $\gamma_+,\gamma_-$ and $\gamma_z$ and hence to the complete positivity of the evolution.
\end{remark}

\begin{figure}[thb]
\centering
\includegraphics[width=0.35\textwidth]{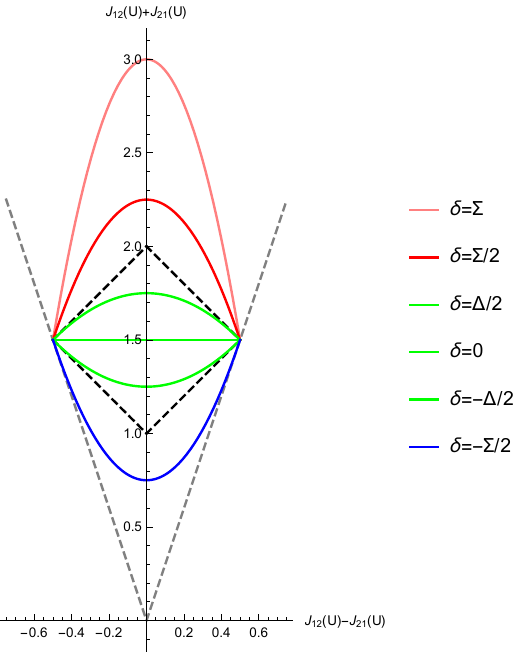}
\includegraphics[width=0.55\textwidth]{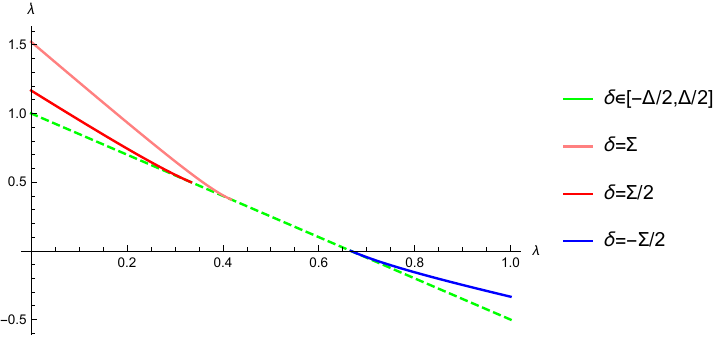}
\caption{Several examples of systems defined by Bloch equations. For all cases we chose $\gamma_+=1$ and $\gamma_-=1/2$. Hence $\Sigma=3/2$ and $\Delta=1/2$. We consider different values of $\gamma_z$, and hence $\delta$.
Left: The space of generators $\mf Q$ for fixed $\Sigma$ and $\Delta$ and different $\delta$, as described in Lemma~\ref{lemma:bloch-sol}. In each case we get a parabolic segment with endpoints $(\pm\Delta,\Sigma)$ and intersecting the $y$-axis at $\Sigma+\delta$. By definition $\delta\geq\Sigma/2$ and hence the parabolas are contained between the lines connecting the endpoints with the origin (gray dashed). This allows us to find the purest achievable state in Corollary~\ref{coro:bloch-purest}.
Right: The optimal derivative $\mu$ for fixed $\Sigma$ and $\Delta$ and different $\delta$, as described in Lemma~\ref{lemma:bloch-upper-bound}. For $\delta\in[-\Delta/2,\Delta/2]$ the optimal derivative $\mu$ is linear (green dashed) and only depends on $\gamma_+$ and $\gamma_-$ (or $\Sigma$ and $\Delta$ equivalently). For other values of $\delta$ the upper bound has to be modified on an interval $I$ defined in the same lemma. Note that this modification does not affect the intersection of $\mu$ with the abscissa, again reflecting Corollary~\ref{coro:bloch-purest}.
}
\label{fig:various-bloch}
\end{figure}

The following results analytically parametrize the space of generators $\mf Q$ and deduce the analytical formula for the optimal derivative function $\mu$. See Figure~\ref{fig:various-bloch} for examples. As a consequence we can also determine the purest stabilizable state.

\begin{lemma}
\label{lemma:bloch-sol}
For the Bloch equations, the space of generators $\mf Q$ is the graph of the parabolic segment
\begin{align*}
f(a) = \Sigma + \delta\big(1-\tfrac{a^2}{\Delta^2}\big),
\quad a\in[-\Delta,\Delta].
\end{align*}
The point $(a,f(a))$ can be obtained using the unitary $\exp(i\pi x\sigma_x)$ satisfying $a=\Delta\cos(2\pi x)$.
\end{lemma}

\begin{proof}
Consider an initial density matrix on the $z$-axis of the Bloch ball. 
Then every density matrix in its unitary orbit can be reached by applying an $x$-rotation followed by a $z$-rotation. 
However, since the Lindblad terms are $z$-symmetric, it suffices to consider only $x$-rotations. That is
\begin{align} \label{eq:bloch-param1}
\mf Q=\{(J_{12}(U)-J_{21}(U),J_{12}(U)+J_{21}(U))
:U=\exp(i\pi x\sigma_x), x\in[0,1/2]\}\,.
\end{align}
Then, evaluating the above expression one obtains the points
\begin{align} \label{eq:bloch-param2}
((\gamma_+-\gamma_-)\cos(2\pi x),\;
\tfrac{3(\gamma_++\gamma_-)}{4}+\gamma_z+\tfrac{1}{4}(\gamma_++\gamma_--4\gamma_z)\cos(4\pi x)), \quad x\in[0,1/2].
\end{align}
Setting $a=\Delta\cos(2\pi x)$ and noting that $\cos(4\pi x)=2(a/\Delta)^2-1$ we obtain the desired points.
\end{proof}

\begin{lemma}
\label{lemma:bloch-upper-bound}
For the Bloch equations, the upper bound $\mu:[0,1]\to\R$ takes the form
\begin{align*}
\mu(\lambda) = 
\begin{cases}
\frac{\Delta^2}{8\delta(1-2\lambda)} + \frac{(\Sigma+\delta)(1-2\lambda)}{2}
\, &\text{ if } \lambda \in I \\
\frac{1}{2} (\Delta+(1-2\lambda)\Sigma) &\text{ if } \lambda\in[0,1]\setminus I,
\end{cases}
\,\text{ where }
I=
\begin{cases}
[0,\tfrac{1}{2}-\tfrac{\Delta}{4\delta}] \,&\text{ if } \delta\geq\frac{\Delta}{2} \\
[\tfrac{1}{2}-\tfrac{\Delta}{4\delta},1] \,&\text{ if } \delta\leq-\frac{\Delta}{2} \\
\emptyset &\text{ otherwise },
\end{cases}
\end{align*}
in particular $1/2 \notin I$.
\end{lemma}

\begin{proof}
Due to Proposition~\ref{prop:polarity} we have to solve the linear maximization $\max_{(a,b)\in\mf Q}\frac{1}{2} (a+(1-2\lambda)b)$.
We see that if $\lambda=1/2$ then the maximum is achieved at $(\Delta,\Sigma)\in\mf Q$ with value $\Delta/2$. 
We have to consider the shape of the parabolic segment depending on $\delta$. 
It is clear that for $\delta=0$ this is just a line segment, for $\delta<0$ it is convex, and for $\delta>0$ it is concave. 
Note also that $f'(\Delta)=\frac{-2\delta}{\Delta}$. Hence if $-\Delta/2<\delta<\Delta/2$ then the maximum will be achieved on $(\Delta,\Sigma)\in\mf Q$ for all $\lambda\in[0,1]$ and so $\mu$ will be the affine linear function $\mu(\lambda)=\frac12(\Delta+(1-2\lambda)\Sigma)$.
This proves the case $I=\emptyset$. 
Now let $\delta\geq\Delta/2$ (the case $\delta\leq-\Delta/2$ is analogous).
Then the parabolic segment is concave. 
For large $\lambda$, the maximum will still be at $(\Delta,\Sigma)\in\mf Q$, but for $\lambda$ small enough, the maximum will be achieved in the interior of the parabolic segment. 
The switching point occurs when the vector $(1,1-2\lambda)$ is orthogonal to $(1,f'(\Delta))$, that is when
$1+(1-2\lambda)(-2\delta/\Delta)=0$ which is equivalent to $\lambda = \frac12-\frac{\Delta}{4\delta}$. 
Hence it remains to determine $\mu$ on $I=[0,\frac12-\frac{\Delta}{4\delta}]$. 
For this we compute
\begin{align} \label{eq:bloch-opt-x}
\tfrac{-2\delta a}{\Delta^2}(1-2\lambda) = -1 \iff a(\lambda)=\tfrac{\Delta^2}{2\delta(1-2\lambda)},
\end{align}
and plugging in we get on $I$ that
$\mu(\lambda) = \tfrac12 (a(\lambda) +(1-2\lambda)f(a(\lambda)))$,
which evaluates to the desired result.
\end{proof}

\begin{corollary}
\label{coro:bloch-purest}
The purest stabilizable state $\lambda^\star$, that is the state satisfying $\mu(\lambda^\star)=0$, is given by
$\lambda^\star = \tfrac12+\tfrac\Delta{2\Sigma}$
and $\{\lambda^\star,1-\lambda^\star\}$ is the spectrum of the fixed point of the system.
\end{corollary}

\begin{proof}
Since $\gamma_z\geq0$ we must have $\delta\geq-\Sigma/2$, and hence 
$f'(\Delta)=\frac{-2\delta}{\Delta}\leq\frac{\Sigma}{\Delta}$,
which shows that the line defined by $(\Delta,\Sigma)\in\mf Q$ is the line giving the purest stabilizable state. Hence
$\tfrac{1}{2} (\Delta+(1-2\lambda^\star)\Sigma)=0 \iff 
\lambda^\star = \tfrac12+\tfrac\Delta{2\Sigma}$,
as desired. 
(This also follows from Lemma~\ref{lemma:heat-baths} and noting that $\gamma_z$ can be set to $0$ without loss of generality.) 
By $z$-symmetry, the diagonal states are invariant, and hence by Brouwer's theorem, there is a fixed point which is diagonal. 
If this fixed point is pure, then $\gamma_-=0$ and we are done. 
Otherwise, the fixed point is unique. 
If we denote the larger eigenvalue of the fixed point by $\lambda$ we obtain that
$-\gamma_-\lambda + \gamma_+(1-\lambda) = 0$
which shows that $\lambda=\tfrac12+\tfrac\Delta{2\Sigma}$, as desired.
\end{proof}

From these results one can deduce the optimal path through the Bloch ball and the corresponding optimal controls in the original control system~\eqref{eq:bilin}, see~\cite[Sec.~6]{Reduced23} for details.
Here we just recover the so called magic plane and steady state ellipsoid from~\cite{Lapert10,Lapert13}.

\begin{lemma}
\label{lemma:magic-plane}
If $\delta\geq\Delta/2$, or equivalently, $\gamma_z\geq\gamma_+/2$, then for $\lambda\in I=[0,\frac12-\frac{\Delta}{4\delta}]$, the points in the Bloch ball achieving the optimal derivative $\mu(\lambda)$ are given by the plane perpendicular to the $z$-axis and passing through the density matrix $\diag(\lambda,1-\lambda)$ with
$\lambda=\frac{1}{2}-\frac{\Delta}{4\delta}$.
If $\lambda\in[\frac12-\frac{\Delta}{4\delta},\frac12]$, then the maximal derivative is reached on the $z$-axis on the same side of the origin as the magic plane.
\end{lemma}

\begin{proof}
If $\delta\geq\Delta/2$, then the interval $I=[0,\frac{1}{2}-\frac{\Delta}{4\delta}]$ lies in $[0,\frac{1}{2}]$. 
For $\lambda\in I$, the upper bound $\mu(\lambda)$ is non-linear. 
We want to find, for each $\lambda\in I$, the density matrices $\rho$ for which the optimal derivative $\mu(\lambda)$ is achieved. 
For this we find $U$ such that the optimal derivative is achieved for $U\diag(\lambda,1-\lambda)U^*$. 
From~\eqref{eq:bloch-opt-x} it follows that for $\lambda\in I$ the optimal point in the space of generators $\mf Q$ has $a=\Delta^2/(2\delta(1-2\lambda))$. 
From~\eqref{eq:bloch-param1} and~\eqref{eq:bloch-param2} it follows that a corresponding unitary is $\exp(i\pi x\sigma_x)$ with $a=\Delta\cos(2\pi x)$. 
Hence the optimal unitary can be found from
$x=\frac{1}{2\pi}\arccos\big(\frac{\Delta}{2\delta(1-2\lambda)}\big)$.
Using the $z$-symmetry of the problem, this shows that for $\lambda\in I$ the optimal derivatives are reached on a plane perpendicular to the $z$-axis and passing through the point $\diag(\lambda,1-\lambda)$ with $\lambda=\frac12-\frac\Delta{4\delta}$.
\end{proof}

\begin{remark} \label{rmk:magic-plane}
Lemma~\ref{lemma:magic-plane} recovers a known result from~\cite{Lapert10,Lapert13}, which was derived using the Pontryagin maximum principle. 
There, the parameters are scaled such that $\Sigma=\Delta$, or equivalently $\gamma_-=0$, and hence $\gamma_+=T_1^{-1}$ and $\gamma_z = (2T_2)^{-1} - (4T_1)^{-1}$. 
The condition from Lemma~\ref{lemma:magic-plane} then becomes $T_1\geq\frac{3}{2}T_2$ and the magic plane intersects the $z$-axis at radius $r=T_2/(2(T_1-T_2))$ in the lower half.
Recall that here the radius of the Bloch sphere is $\tfrac12$.
\end{remark}

\begin{lemma} \label{lemma:bloch-stab}
The set of stabilizable states in the Bloch disk is an ellipsoid in the upper halfspace which is rotationally symmetric around the $z$-axis and has a vertical semiaxis of length $\tfrac{\Delta}{4\Sigma}$ and horizontal semiaxes of length $\tfrac{\Delta}{4\sqrt{\Sigma(\Sigma+\delta/2)}}$.
\end{lemma}

\begin{proof}
By Lemma~\ref{lemma:stab-ball} we obtain the parametrization in polar coordinates
$r(\theta,\phi)=\frac{\Delta}{2}\frac{\cos(\theta)}{\Sigma+\delta\sin(\theta)^2}$.
The result follows from Lemma~\ref{lemma:up-ellipsoid}.
\end{proof}

\begin{remark}
This recovers another known result from~\cite{Lapert13}, with parameters rescaled as in Remark~\ref{rmk:magic-plane}. 
Then the ellipse has vertical semiaxis length $1/4$, i.e. it touches the north pole, and horizontal semiaxis length $\sqrt{T_2/(2T_1)}$.
Again recall that here the radius of the Bloch sphere is $\tfrac12$.
\end{remark}

\section{Technical computations} \label{app:calc}

\subsection{Detailed computation for Proposition~\ref{prop:general-parametrization}}
\label{app:parametrization}

The unitary $U_{x,z}=\exp(\iu\pi z\sigma_z)\exp(\iu\pi x\sigma_x)$ has the form
\begin{align*}
U_{x,z}=\begin{pmatrix}
\alpha c&\iu\alpha s\\
\iu\alpha^*s&\alpha^*c
\end{pmatrix},\
\text{ and hence }\
U_{x,z}^*=
\begin{pmatrix}
\alpha^* c&-\iu\alpha s\\
-\iu\alpha^*s&\alpha c
\end{pmatrix},
\end{align*}
where $\alpha=e^{\iu\pi z}$ and $s=\sin(\pi x)$ and $c=\cos(\pi x)$. For a traceless Lindblad term
\begin{align*}
V=\begin{pmatrix}u&v\\w&-u\end{pmatrix}    
\end{align*}
where $u,v,w\in\C$ are arbitrary we immediately obtain
\begin{align*}
(U_{x,z}^* V U_{x,z})_{12} &= 2\iu cs u  + (\alpha^*)^2c^2v  + \alpha^2s^2w \\
(U_{x,z}^* V U_{x,z})_{21} &= -2\iu cs u + (\alpha^*)^2s^2v + \alpha^2c^2w\,,
\end{align*}
and thus
\begin{align*}
J_{12}(U_{x,z})-J_{21}(U_{x,z}) 
&=
|2\iu cs u + (\alpha^*)^2c^2v + \alpha^2s^2w|^2
-
|-2\iu cs u + (\alpha^*)^2s^2v + \alpha^2c^2w|^2
\\&=
(c^2 - s^2)(|v|^2 - |w|^2) 
+
2\Re(2\iu cs u(\alpha^2v^*+(\alpha^*)^2w^*))
\\&=
(|v|^2-|w|^2)\cos(2\pi x)+
2\sin(2\pi x)\Im(e^{-2\pi\iu z}(uv^*-wu^*))
\end{align*}
where the first term stems from the norm squared terms and the rest from the cross terms. Next we find
\begin{align*}
J_{12}(U_{x,z})+J_{21}(U_{x,z}) 
&=
|2\iu cs u + (\alpha^*)^2c^2v + \alpha^2s^2w|^2
+
|-2\iu cs u + (\alpha^*)^2s^2v + \alpha^2c^2w|^2
\\&= 
8c^2s^2|u|^2
+ (c^4+s^4)(|v|^2+|w|^2)
\\&\quad + 4c^2s^2\Re(\alpha^4vw^*)
 - 4cs(c^2-s^2)\Im(u(\alpha^2 v^*-(\alpha^*)^2 w^*)),
\end{align*}
where the first two terms stem from the norm squared terms and the remaining ones from the cross terms. Using the trigonometric identities $\cos(\pi x)^4 + \sin(\pi x)^4=1-\sin(2\pi x)^2/2$ and $\cos(\pi x)^2 - \sin(\pi x)^2 = \cos(2\pi x)$ and $\cos(\pi x)  \sin(\pi x)=\sin(2\pi x)/2$ we obtain
\begin{align*}
J_{12}(U_{x,z})+J_{21}(U_{x,z}) 
&=
|v|^2+|w|^2 
+ \sin(2\pi x)^2\big(2|u|^2-\tfrac{|v|^2+|w|^2}{2}\big)
\\&\quad + \sin(2\pi x)^2 \Re(\alpha^4vw^*)
 - 2\sin(2\pi x)\cos(2\pi x)\Im(\alpha^2(uv^*+u^*w)).
\end{align*}
If we now consider a finite family of traceless Lindblad terms
\begin{align*}
V_k=\begin{pmatrix}u_k&v_k\\w_k&-u_k\end{pmatrix},
\end{align*}
and define $\Sigma=\sum_{k=1}^r |v_k|^2+|w_k|^2$, and $\Delta=\sum_{k=1}^r |v_k|^2-|w_k|^2$, and $\delta=\sum_{k=1}^r 2|u_k|^2-(|v_k|^2+|w_k|^2)/2$, as well as $r_1e^{\iu\phi_1}=\iu\sum_{k=1}^r v_kw_k^*$ and $r_2e^{\iu\phi_2}=2\sum_{k=1}^r(u_kv_k^*+u_k^*w_k)=$ then we obtain
\begin{align*} 
J_{12}(U_{x,z})-J_{21}(U_{x,z}) 
&= \Delta\cos(2\pi x)+
\sin(2\pi x)\textstyle\sum_{k=1}^r\Im(e^{-2\pi\iu z}(u_kv_k^*-w_ku_k^*)) \\
J_{12}(U_{x,z})+J_{21}(U_{x,z}) 
&= \Sigma + \delta\sin(2\pi x)^2 + r_1\sin(2\pi x)^2\sin(4\pi z+\phi_1)
\\&\quad - r_2\sin(2\pi x)\cos(2\pi x)\sin(2\pi z+\phi_2),
\end{align*}
since $\Re(\alpha^4v_kw_k^*)=\Re(-\iu r_1 e^{\iu (4\pi z+\pi_1)})=r_1\sin(4\pi z+\pi_1)$ and since $2\Im(\alpha^2(u_kv_k^*+u_k^*w_k))=r_2\sin(2\pi z+\phi_2)$. Finally, since $\sum_{k=1}^r[V_k,V_k^*]$ is diagonal if and only if $\sum_{k=1}^r u_kv_k^*=\sum_k w_ku_k^*$, we obtain in this case that
\begin{align*}
J_{12}(U_{x,z})-J_{21}(U_{x,z}) 
&= \Delta\cos(2\pi x).
\end{align*}

\subsection{Study of the function $G_\phi$}
\label{app:g-phi}

\begin{figure}[htp]
\centering
\includegraphics[width=0.45\textwidth]{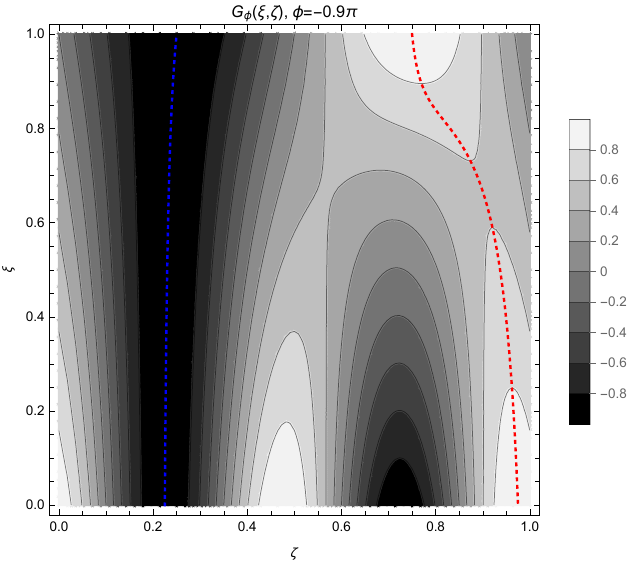}
\includegraphics[width=0.45\textwidth]{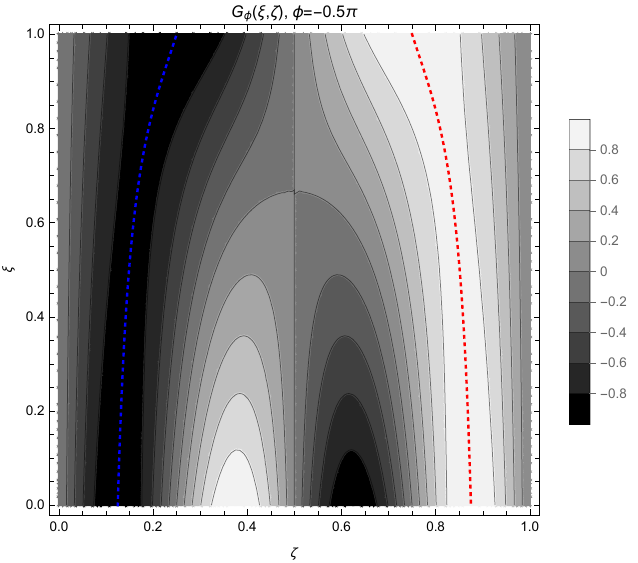}
\includegraphics[width=0.45\textwidth]{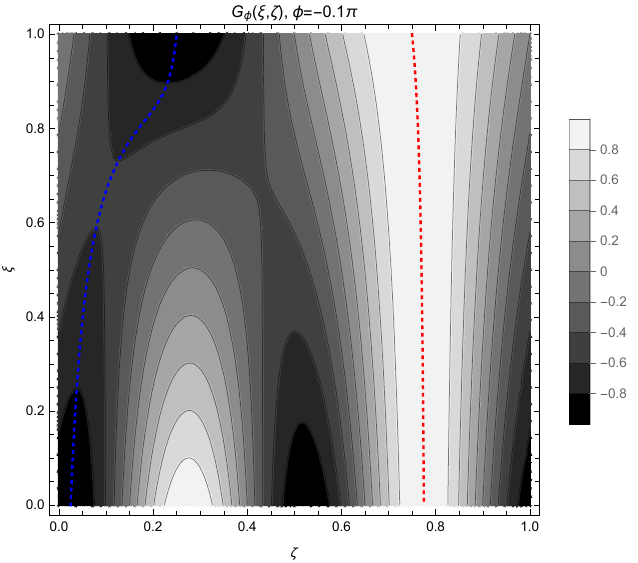}
\includegraphics[width=0.45\textwidth]{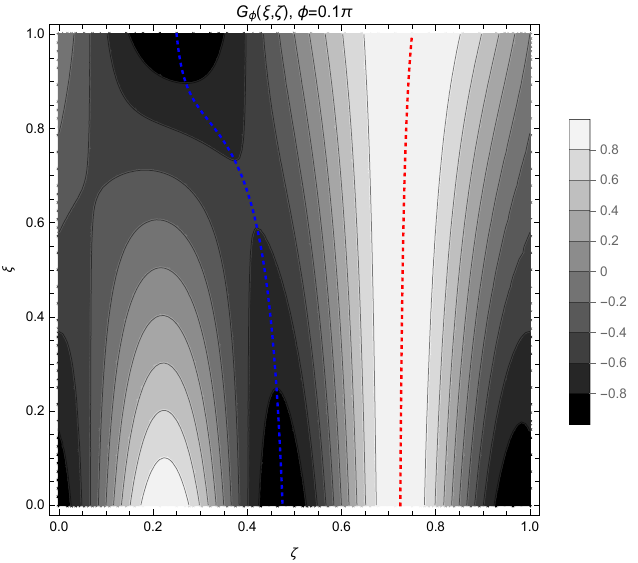}
\includegraphics[width=0.45\textwidth]{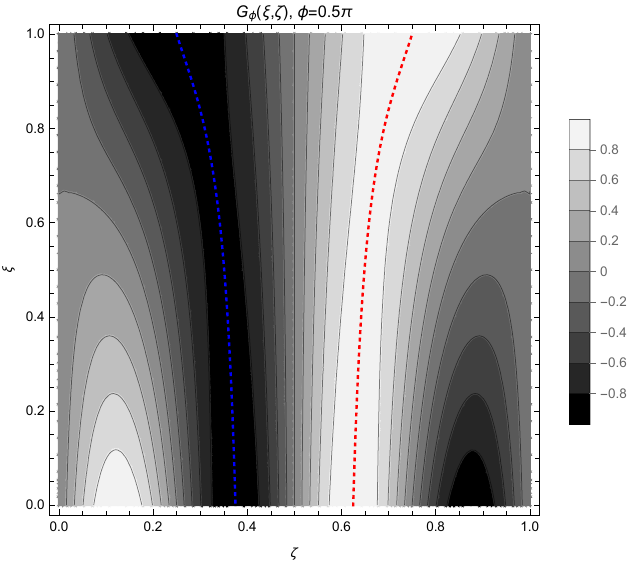}
\includegraphics[width=0.45\textwidth]{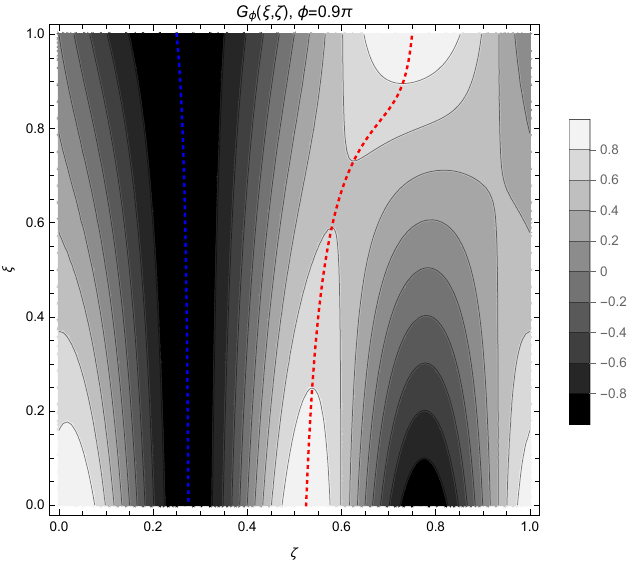}
\caption{Contour plots of the function $G_\phi(\xi,\zeta)$ for different values of $\phi$ together with the maximizers (red) and minimizers (blue) as a function of $\xi$.} 
\label{fig:contour-g}
\end{figure}

In~\eqref{eq:G-phi} we defined the function
\begin{align*}
G_\phi(\xi,\zeta)=
(1-\xi)\sin(4\pi \zeta+\phi-\pi/2)-\xi\sin(2\pi\zeta).
\end{align*}
for $\xi,\zeta\in[0,1]$ and $\phi\in(-\pi,\pi]$. See Figure~\ref{fig:contour-g} for plots of this function for different values of $\phi$.  

\begin{lemma} \label{lemma:G-phi-opt}
Let $\phi\in(-\pi,0)\cup(0,\pi)$ and define the function
\begin{align*}
\xi^\star(\zeta) = \frac{1}{1+\frac{\cos(2\pi\zeta)}{2\cos(4\pi\zeta+\phi-\pi/2)}},
\end{align*}
and the intervals
\begin{align*}
I^+=\begin{cases}
[3/4-\phi/(4\pi),3/4] \text{ if } \phi>0\\
[3/4,3/4-\phi/(4\pi)] \text{ if } \phi<0,
\end{cases} \quad
I^-=\begin{cases}
[1/4,1/4+(\pi-\phi)/(4\pi))] \text{ if } \phi>0\\
[-\phi/(4\pi),1/4] \text{ if } \phi<0.
\end{cases}
\end{align*}
Let $\xi^+$ and $\xi^-$ denote the restrictions of $\xi^\star$ to $I^+$ and $I^-$ respectively. These functions are bijective onto $[0,1]$, and it holds that
\begin{align*}
(\xi^+)^{-1}(\xi) = \underset{\tilde\zeta\in[0,1]}{\argmax}\ G_\phi(\xi,\tilde \zeta), \qquad
(\xi^-)^{-1}(\xi) = \underset{\tilde\zeta\in[0,1]}{\argmin}\ G_\phi(\xi,\tilde \zeta),
\end{align*}
for all $\xi\in[0,1]$. Here the $\argmax$ and $\argmin$ are unique for all $\xi\in(0,1]$ and for $\xi=0$ there is a spurious second solution which we omit.
\end{lemma}

The Lemma above runs into problems when $\phi\in\{0,\pi\}$, since in this case the functions $\xi^+$ and $\xi^-$ are not surjective to $[0,1]$ anymore. However it turns out that in these cases we can compute the desired maximizers and minimizers explicitly.

\begin{lemma} \label{lemma:G-phi-opt-integral}
If $\phi=0$ then we define $\zeta^+(\xi)\equiv3/4$ and
\begin{align*}
\zeta^-(\xi) = \begin{cases}
\frac{1}{4} \pm (\frac{1}{2\pi}\arcsin{\frac{1}{4(1/\xi-1)}}-\frac{1}{4}) &\text{ if } \xi\in[0,4/5] \\
\frac{1}{4} &\text{ if } \xi\in[4/5,1],
\end{cases}
\end{align*}
then $\zeta^+(\xi)$ is the unique maximizer of $G_0(\xi,\cdot)$ for all $\xi\in[0,1]$ and the two possibilities for $\zeta^-(\xi)$ are the only minimizers for $G_0(\xi,\cdot)$. Similarly, for $\phi=\pi$ we have $\zeta^-(\xi)\equiv1/4$ and
\begin{align*}
\zeta^+(\xi) = \begin{cases}
\frac{3}{4} \pm (\frac{1}{2\pi}\arcsin{\frac{1}{4(1/\xi-1)}}-\frac{3}{4}) &\text{ if } \xi\in[0,4/5] \\
\frac{3}{4} &\text{ if } \xi\in[4/5,1].
\end{cases}
\end{align*}
\end{lemma}

\subsection{Ellipses} \label{app:ellipse}

In this section we derive a useful polar coordinate parametrization of ellipses passing through the origin.

For $a,b>0$ and $\phi_0\in(-\pi,\pi]$ we define the following parametrization of an ellipse:
$$
E_{a,b,\phi_0}(\phi)=\begin{pmatrix}
a(\cos(\phi)-\cos(\phi_0))\\b(\sin(\phi)-\sin(\phi_0))
\end{pmatrix},
\qquad \phi\in(-\pi,\pi].
$$
This is simply an ellipse with semiaxes of length $a$ and $b$, translated such that it intersect the origin when $\phi=\phi_0$, and the center of the ellipse is given by $(x_0,y_0)=(-a\cos(\phi_0),-b\sin(\phi_0))$.

\begin{lemma}
In polar coordinates $(r,\psi)$, where $r$ is allowed to be negative, we can parametrize the ellipse $E_{a,b,\phi_0}$ as
\begin{align} \label{eq:ellipse-polar}
r(\psi) = -2 \frac{\frac{\cos(\phi_0)\cos(\psi)}{a}+\frac{\sin(\phi_0)\sin(\psi)}{b}}{\frac{\cos(\psi)^2}{a^2}+\frac{\sin(\psi)^2}{b^2}}, \quad \psi \in (-\pi/2,\pi/2]. 
\end{align}
\end{lemma}

\begin{proof}
Let $\psi \in (-\pi/2,\pi/2)$ be given. The corresponding slope is $m=\tan(\psi)$. We are looking for the intersections of the line $y=mx$ and the ellipse
$$
\left(\frac{x-x_0}{a}\right)^2+\left(\frac{y-y_0}{b}\right)^2=1.
$$
Note that the origin is always in this intersection. If the line is tangent to the ellipse, this is the only intersection, otherwise there exists exactly one more. We plug in and using $x\neq0$ we find 
$$
\left(\frac{1}{a^2}-\frac{m^2}{b^2}\right) x^2 + \left(\frac{-2x_0}{a^2}-\frac{2my_0}{b^2}\right) + \left(\frac{x_0}{a}\right)^2 +  \left(\frac{y_0}{b}\right)^2=1
$$
and using the definition of $(x_0,y_0)$ this implies that 
$$
x = 2\frac{\frac{x_0}{a}+\frac{my_0}{b}}{\frac{1}{a^2}+\frac{m^2}{b^2}}.
$$
The distance $r$ of the intersection to the origin can be found by computing
$$
r^2=(1+m^2)x^2 = 4 \left(\frac{\frac{x_0\cos(\psi)}{a}+\frac{y_0\sin(\psi)}{b}}{\frac{\cos(\psi)^2}{a^2}+\frac{\sin(\psi)^2}{b^2}}\right)^2,
$$
and since by the choice of the range of $\psi$ it holds that $\sign(r)=\sign(x)$, we see that
$$
r=2 \frac{\frac{x_0\cos(\psi)}{a}+\frac{y_0\sin(\psi)}{b}}{\frac{\cos(\psi)^2}{a^2}+\frac{\sin(\psi)^2}{b^2}},
$$
and by continuity this formula remains true for $\psi=\pm\pi/2$, and this concludes the proof.
\end{proof}

Note that the ellipse can furthermore be rotated around the center by shifting the angular coordinate, i.e., $\psi\mapsto r(\psi-\theta)$. In fact there is a unique angle $\theta^\star$ (modulo $2\pi$), such that the ellipse lies in the upper halfplane, and the parametrization takes on a simplified form.

\begin{lemma} \label{lemma:ellipse-eq}
Let $a,b>0$ and $\phi\in(-\pi,\pi]$ be given and let $E_{a,b,\phi}$ be the corresponding ellipse, with polar parametrization $r(\psi)$ as in~\eqref{eq:ellipse-polar}. Then 
$$
\theta^\star = \arctan(b\cos(\phi),a\sin(\phi)) \in(-\pi,\pi]
$$
is the unique angle in $(-\pi,\pi]$ such that $\psi\mapsto r(\psi-\theta^\star)$ takes image in the upper halfplane. Moreover it holds that
\begin{align} \label{eq:ellipse-polar2}
r(\psi-\theta^\star) = \frac{\sigma \sin(\psi)}{\alpha + \beta\cos(\psi)^2 + \gamma\sin(\psi)\cos(\psi)}
\end{align}
where
\begin{alignat*}{3}
\alpha &= \left(\frac{\cos(\theta^\star)}{b}\right)^2 + \left(\frac{\sin(\theta^\star)}{a}\right)^2 & \qquad
\beta  &= \left(\frac{1}{a^2}-\frac{1}{b^2}\right)\cos(2\theta^\star) \\
\gamma &= \left(\frac{1}{a^2}-\frac{1}{b^2}\right)\sin(2\theta^\star)  & \qquad
\sigma &= -2\left(\frac{1}{a}\cos(\phi)\sin(\theta^\star) + \frac{1}{b}\sin(\phi)\cos(\theta^\star)\right).
\end{alignat*}
\end{lemma}

\begin{proof}
If we expand the numerator of $r(\psi-\theta)$ we obtain a linear combination of $\cos(\psi)$ and $\sin(\psi)$. For the ellipse to lie in the upper or lower halfplane, the coefficient of $\cos(\psi)$ must be zero. This gives the condition
$$
\tan(\theta)=\frac{b \cos(\phi)}{a \sin(\phi)}.
$$
This defines $\theta$ only modulo $\pi$. To make sure that the ellipse is in the upper halfplane, note that the function which maps $(a,b,\phi)$ to the function $\psi\mapsto r(\psi-\theta^\star)$ is continuous, and the parameter space of $(a,b,\phi)$ is connected, hence the ellipse will always lie in the same halfspace. Thus it suffices to check one ellipse, e.g., $a=b=1$ and $\phi=0$. The remainder of the proof is a straightforward computation using elementary trigonometric identities which we will omit.
\end{proof}

Next we want to find $a,b,$ and $\phi$ given a parametrization as in~\eqref{eq:ellipse-polar2}. We start with a special case.

\begin{lemma} \label{lemma:ellipse-params}
Let $\alpha,\beta,\gamma,\sigma\in\R$ with $\alpha\neq0$ and consider the parametrization 
$$
r(\psi) = \frac{\sigma\sin(\psi)}{\alpha+\beta\cos(\psi)^2+\gamma\sin(\psi)\cos(\psi)}.
$$
Without loss of generality we may assume that $\beta^2+\gamma^2=1$\footnote{%
If $\beta=\gamma=0$ then the parametrization is that of a circle of diameter $\sigma/\alpha$.
} and $\sigma>0$.
Then, this is the parametrization of an ellipse if and only if $2\alpha+\beta\notin[-1,1]$, which corresponds to the denominator being non-zero for all $\psi$. 
The ellipse lies in the upper halfplane if and only if $2\alpha+\beta>1$.
In this case, the ellipse is exactly $E_{a/s,b/s,\phi}$ where 
\begin{align*}
\theta =\frac{1}{2}\arctan(\gamma,\beta),\quad
a      = \frac{\sigma}{\sqrt{|\alpha-\sin^2(\theta)|}},\quad 
b      = \frac{\sigma}{\sqrt{|\alpha+\cos^2(\theta)|}} \\
\phi   =  \arctan(b\cos(\theta), a\sin(\theta)),\quad
s = \frac{1}{a}\cos(\phi)\sin(\theta) + \frac{1}{b}\sin(\phi)\cos(\theta).
\end{align*}
\end{lemma}

\begin{proof}
This can be verified by plugging the values into Lemma~\ref{lemma:ellipse-eq}.
\end{proof}

\subsection{Upright ellipsoids} \label{sec:upright}

Next we address the case of ellipsoids in three dimensions.
We will only consider ellipsoids in the upper half space intersecting the origin and whose center lies on the $z$-axis, and we will say that they are \emph{upright}.
As a consequence of Lemma~\ref{lemma:ellipse-params} we find the parametrization of ellipses which are upright in the analogous sense.

\begin{corollary}
Let $\alpha,\beta,\sigma\in\R$ with $\sigma>0$ and $\alpha\neq0$ as well as $\alpha>\max(0,-\beta)$. Then the curve given in polar coordinates by
$$
r(\psi)=\frac{\sigma\sin(\psi)}{\alpha+\beta\cos(\psi)^2}
$$
parametrizes the upright ellipse $E_{a,b,0}$ with
$$
a=\frac{\sigma}{2|\alpha|},\quad b=\frac{\sigma}{2\sqrt{|\alpha(\alpha+\beta)|}}.
$$
\end{corollary}

\begin{lemma} \label{lemma:up-ellipsoid}
Let $\alpha,\beta,\eta,\sigma\in\R$ with $\sigma>0$ and $\alpha\neq0$.
Then the surface parametrized by
$$r(\theta,\phi) = \frac{\sigma\cos(\theta)}{\alpha + (\beta+\eta\cos(2\phi))\sin(\theta)^2}$$
is an upright ellipsoid with axes
$$
a=\frac{\sigma}{2\sqrt{|\alpha(\alpha+\beta+\eta)|}}, \quad
b=\frac{\sigma}{2\sqrt{|\alpha(\alpha+\beta-\eta)|}}, \quad 
c=\frac{\sigma}{2|\alpha|}.
$$
\end{lemma}

\printbibliography[heading=bibintoc]

\end{document}